\documentclass[journal]{IEEEtran}
\usepackage{cite}
\usepackage{amsmath,amsthm,amssymb,amsfonts}

\usepackage{algorithmic}
\usepackage{mathtools}
\usepackage{graphicx}
\usepackage{balance}
\usepackage{orcidlink}

\theoremstyle{plain}
\newtheorem{proposition}{Proposition}[section]
\newtheorem{theorem}{Theorem}[section]
\theoremstyle{definition}
\newtheorem{definition}{Definition}[section]
\theoremstyle{remark}
\newtheorem{remark}{Remark}[section]

\def\BibTeX{{\rm B\kern-.05em{\sc i\kern-.025em b}\kern-.08em
    T\kern-.1667em\lower.7ex\hbox{E}\kern-.125emX}}

\begin{document}

\title{Luré-Postnikov Stability Analysis of Closed-Loop Control Systems with Gated Recurrent Neural Network-based Virtual Sensors}

\author{%
  E.~Hilgert\,\orcidlink{0009-0004-4554-8499},%
  ~and
  A.~Schwung\,\orcidlink{0000-0001-8405-0977}%
  \thanks{Manuscript received April~28,~2025; revised~—; accepted~—.
           Corresponding author:~E.~Hilgert (e\,mail: hilgert.er@gmail.com).}%
  \thanks{E.~Hilgert is with the Department of Technology and Systems, Graduate School for Applied Research in North Rhine‑Westphalia, Bochum, Germany.}%
  \thanks{A.~Schwung is with the Department of Automation Technology and Learning Systems, South Westphalia University of Applied Sciences, Soest, Germany.}%
  \thanks{This work has been submitted to the IEEE for possible publication. Copyright may be transferred without notice, after which this version may no longer be accessible.}%
}

\maketitle
\thispagestyle{empty}
\pagestyle{empty}

\begin{abstract}
This article addresses certification of closed‑loop stability when a virtual‑sensor based on a gated recurrent neural network operates in the feedback path of a nonlinear control system. The Hadamard gating used in standard GRU/LSTM cells is shown to violate the Luré–Postnikov Lyapunov conditions of absolute‑stability theory, leading to conservative analysis. To overcome this limitation, a modified architecture—termed the Luré–Postnikov gated recurrent neural network (LP‑GRNN)—is proposed; its affine update law is compatible with the Luré–Postnikov framework while matching the prediction accuracy of vanilla GRU/LSTM models on the NASA CMAPSS benchmark. Embedding the LP‑GRNN, the plant, and a saturated PI controller in a unified standard nonlinear operator form (SNOF) reduces the stability problem to a compact set of tractable linear matrix inequalities (LMIs) whose feasibility certifies global asymptotic stability. A linearized boiler case study illustrates the workflow and validates the closed‑loop performance, thereby bridging modern virtual‑sensor design with formal stability guarantees.
\end{abstract}

\IEEEpubidadjcol

\section{Introduction}\label{sec:intro}
\IEEEPARstart{A}{ccurate} sensing is a main component of modern control systems, yet in many industrial scenarios physical measurement of critical variables is impractical or costly. Complex installation requirements, harsh environments, or simply the cost of high-quality instrumentation might leave important states unobserved. This challenge has led to the development of virtual sensors (or soft sensors), which estimate unmeasured quantities indirectly from other available signals \cite{Jiang2021ARO}. For clarity throughout the paper, we distinguish between a \emph{virtual sensor} and a \emph{state observer}: while state observers typically reconstruct internal system states, our work focuses on virtual sensors that estimate plant outputs $y_k$. Various approaches exist for developing virtual sensor models \cite{Jiang2021ARO}. Examples include first principle models \cite{Ruderman2016OnSA}, observer-based approaches \cite{Yin2012DatadrivenAO}, filtering techniques \cite{Canale2008ASO}, multivariate statistical analysis methods \cite{Godoy2011MultivariateSM} and machine learning-based approaches - the scope of this paper. The advances in machine learning (ML) have enhanced the capabilities of such virtual sensors, which led to a widespread usage of ML-based virtual sensors across domains, e.g.\cite{atkin,en15155743,Pattnaik2021MachineLB,Kim2021PredictionOE,SOK2024123224,icsbitirici2023lstm}.   In particular, recurrent neural networks (RNNs) like Gated Recurrent Units (GRUs) and Long Short-Term Memory (LSTM) networks are increasingly used for these tasks, thanks to their ability to capture temporal dependencies in time series data.

The development of ML-based virtual sensors comes with usecases within control loops, which introduces new challenges in terms of reliability and stability of such systems. Their accuracy may drop outside the training data distribution and they may exhibit dynamic behavior. Therefore, a controller could be misled, potentially destabilizing the closed-loop system \cite{GonzlezHerbn2024AssessmentAD}. To address these concerns, recent research has begun to combine control-theoretic analysis with machine learning design. Broadly, two complementary approaches have emerged: empirical validation - such as \cite{Canale2008ASO} and \cite{Sieberg2022EnsuringTR} - and formal stability analysis - e.g. \cite{Ruderman2016OnSA}, \cite{Zhirabok2022VirtualSD} and \cite{10394670}. ML models are not only implemented as virtual sensors, but as controllers, plant models or in combination. Thus, the effects of machine learning models in controlled systems is highly relevant and well-researched, see \cite{Moe2018MachineLI}. A variety of formal stability proofs do exist for machine learning models in controlled systems, as outlined in Section \ref{sec:work}. Despite this progress, significant gaps remain. 

To the best of our knowledge, existing stability analysis in control systems of gated recurrent neural networks (RNNs) is currently limited by simplified RNN variants or overly restrictive assumptions. While there exist nonconservative stability methods, the combination of those with nonlinear gating mechanisms in standard GRU and LSTM architectures pose significant challenges.  Our research bridges this gap by making a gated recurrent architecture compatible to a nonconservative Luré-Postnikov based stability method.

Our main contributions can be summarized as follows:
\begin{itemize}
    \item We analyze the stability of nonlinear control systems - consisting of a nonlinear plant / controller combination and a feedback provided by a virtual sensor based on gated recurrent neural network. We base the analysis on Luré Type Systems and the standard operator normal form (SNOF).
    \item We show that the Hadamard gating mechanism in gated recurrent architectures like GRU and LSTM are not suitable for transformations into the SNOF.
    \item We therefore propose an adjusted GRU architecture which allows for a SNOF transformation and stability analysis and show empirically, that the novel architecture is on par to vanilla GRU and LSTM architectures on a benchmark data set.
    \item We provide a formal proof of stability for the control system including the proposed virtual sensor architecture and show the applicability of the approach on a numerical example.
\end{itemize}

The remainder of this article is organized as follows.
Section \ref{sec:work} reviews existing techniques for guaranteeing stability when machine-learning models appear in feedback loops and show the limitations for gated RNNs.
Section \ref{sec:pre} recaps absolute-stability theory and the Standard Nonlinear Operator Form (SNOF) - the foundation of our analysis - and Section \ref{sec:hadamard} shows how absolute stability theory and Hadamard gating in vanilla GRU and LSTM cells are not compatible. We overcome this issue by introducing the LP-GRNN in Section \ref{sec:LP-GRNN}. The LP-GRNN is a modified gated-recurrent architecture that works with absolute stability theory.
Section \ref{sec:stability} develops a holistic transformation of the plant-controller–sensor loop into a single SNOF and derives LMI conditions for global asymptotic stability. In Section \ref{sec:example} we validate the LP-GRNN architecture with a benchmark comparison for a timeseries prediction task and show the stability analysis method on a boiler-plant case study. Finally, Section \ref{sec:conclusion} summarizes the contributions and outlines potential future work.

\section{Related Work}\label{sec:work}

We discuss related work on analyzing the stability of machine learning models in control systems’ feedback loops, with various approaches addressing this challenge.

\subsection{State-Observer Methods} Observer-based methods integrate neural networks into state estimation to recover unmeasured dynamics. These methods usually modify the network significantly to guarantee exponential error convergence using frameworks like Luré-Postnikov and Input-to-State Stability. For example, \cite{10098555} proposed a feedforward residual neural network observer with fal function activations, but their method is limited by conservative LMI conditions and the loss of nonlinear activation details due to global sector bounds. Meanwhile, \cite{10376425} introduced an adaptive LSTM-based observer for nonlinear systems, employing online weight tuning driven by Lyapunov stability. Although their Lyapunov analysis ensures error convergence and overall stability, the weight adaptation law—relying on a user-selected gain matrix and high-dimensional Jacobian calculations—may be highly problem-sensitive, computationally intensive, and prone to numerical issues, challenging its practical online adaptability.

\subsection{Input-to-State Stability (ISS, \texorpdfstring{$\delta$ISS}{deltaISS})} Input-to-State Stability (ISS) ensures that if external disturbances are bounded, the system’s states remain bounded. Incremental stability goes further by guaranteeing that trajectories starting from different initial conditions converge over time. An example for this stability approach is \cite{BONASSI2021105049}. The authors derive sufficient conditions for both ISS and incremental ISS ($\delta$ISS) of GRU networks. Practically, the results allow one to check if a trained GRU satisfies stability, or to impose these conditions during training to directly learn a GRU that is ISS and $\delta$ISS. These conditions are conservative due to the use of worst-case norm bounds and global Lipschitz constraints. Another interesting approach is presented by \cite{10298638}, where a $\delta$ISS condition for specific RNNs is proposed. They use the bounded norms of the weights in combination with the Lipschitz constant to ensure contraction. However, this limits the applicability, as the approach only works for specific contracting implicit RNNs, excluding relevant architectures like GRU and LSTM. There also exists research on LSTM architectures, notably \cite{pmlr-v120-bonassi20a}, where stability verification is achieved by conservatively constraining the weights and biases. A different method uses IQCs for stability analysis, as in \cite{9683530}. The benefit of this approach is its reduced conservatism; however, it is limited to vanilla RNNs without gating mechanisms and with fixed linear state updates.

\subsection{Luré-Postnikov Framework and the Standard Nonlinear Operator Form (SNOF)} The Luré-Postnikov framework models a system as an LTI block in feedback with a static nonlinearity. It constructs a Lyapunov function that captures both the energy of the linear part and the bounded nonlinearity (via sector conditions), leading to LMI conditions that guarantee global stability. Building on this idea, \cite{kim2009robust} utilized slope restriction bounds to formulate a less conservative stability theorem compared to classical criteria like the Circle or Tsypkin criterion, as well as those by \cite{Haddad1994AbsoluteSC, Kapila1994AME, Park2002StabilityCO}. Based on the theorem of \cite{kim2009robust} a Standard Nonlinear Operator Form (SNOF) was introduced by \cite{Kim2018StandardRA}. In SNOF, basic neural network architectures are reformulated such that the system is expressed as a linear dynamic operator in feedback with a static, memoryless nonlinearity. This transformation enables the application of Luré-type stability analysis to neural networks by ensuring that the nonlinearities meet specific sector and slope-restriction conditions. The SNOF transformation has been successfully applied to several architectures, including the Neural State Space Model (NSSM), Dynamic Recurrent Neural Networks (DRNN), and the Global Input-Output Model (GIOM). In addition, \cite{Nguyen2021RobustCT} extended these ideas to ResNets. However, these transformations rely on the assumption that the network’s operations are composed of an affine linear operator followed by a static nonlinearity. Standard gated recurrent architectures like GRU and LSTM employ multiplicative (Hadamard) gating mechanisms. These gates—essential for controlling information flow and maintaining long-term dependencies—introduce element-wise multiplications between the candidate state and gating signals. Such multiplicative interactions break the affine structure required for a direct SNOF transformation and is not compatible with the Lyapunov function stability approach, which will be further explained in Section \ref{sec:hadamard}. 

\subsection{Alternative Gating Mechanism}
A number of alternative gating strategies have been explored, including variants that merge or simplify gating. For example, the Minimal Gated Unit (MGU) proposed by \cite{zhou2016minimal} reduces complexity by merging the input and forget mechanisms into a single gate, thereby lessening the number of multiplicative interactions. Similarly, several studies have investigated coupling input and forget gates in LSTM architectures to reduce parameter count and simplify dynamics—see, e.g., \cite{7508408} and \cite{jozefowicz2015empirical}. There are also innovative gate adaptations like those in \cite{cheng2020refined} and \cite{genet2025siggate}. Nevertheless, even these simplified or alternative gating mechanisms generally remain non-compatible with the Luré framework. They still involve some form of multiplicative nonlinearity that does not readily decompose into a purely affine linear operator and a static nonlinearity, an issue we elaborate in Section \ref{sec:hadamard}.

\subsection{Summarized Research Gap} 
 While there are interesting approaches addressing the stability of machine-learning models in control systems, to our knowledge there are no nonconservative stability proofs for gated recurrent networks that do not restrict either the architecture or its nonlinearities. Existing works either transform feed-forward or simpler recurrent architectures into SNOF or impose conservative conditions on GRU/LSTM weights and biases. Our work builds upon the approach of \cite{Kim2018StandardRA} by reformulating the GRU’s gating mechanism, thereby making it transformable into the SNOF. This novel adjustment allows for stability analysis—either for the neural network in isolation or when integrated into a control loop with plant dynamics as a virtual sensor.

\section{Preliminaries}\label{sec:pre}
This section introduces the concepts and notation used throughout the paper. Specifically, we concentrate on absolute stability theory, standard nonlinear operator forms and gated recurrent neural networks. 

\subsection{Absolute Stability Theory}
Our stability analysis is based on the absolute stability theory of Lur\'e-Postnikov, which addresses the stability of systems structured as a linear time-invariant (LTI) block in feedback with a nonlinear element. A typical discrete-time Lur\'e system is given by
\begin{equation}\label{eq:lure_system}
\begin{aligned}
x_{k+1} &= A\,x_k + B\,\phi(q_k),\\
q_k     &= C\,x_k ,
\end{aligned}
\end{equation}

\noindent
where $x_k\in\mathbb R^{n}$ is the state, 
$q_k\in\mathbb R^{p}$ is the nonlinear intermediate vector, 
$A\in\mathbb R^{n\times n}$, $B\in\mathbb R^{n\times p}$ and 
$C\in\mathbb R^{p\times n}$ are constant matrices,  
and $\phi:\mathbb R^{p}\!\to\!\mathbb R^{p}$ is a component‑wise static non‑linearity.

This formulation is extended via a Linear Fractional Transformation (LFT) to a Standard Nonlinear Operator Form (SNOF) \cite{Kim2018StandardRA} \cite{Jeon2022CompactNN}, which allows for dynamic, multivariate nonlinearities to be modeled and is expressed as:

\begin{equation}\label{eq:SNOF}
\Sigma_k:\;
\left\{
\begin{aligned}
x_{k+1} &= A\,x_k + B_p\,p_k + B_u\,u_k,\\
q_k     &= C_q\,x_k + D_{qp}\,p_k + D_{qu}\,u_k,\\
y_k     &= C_y\,x_k + D_{yp}\,p_k + D_{yu}\,u_k,
\end{aligned}
\right.
\qquad
p_k = \Gamma(q_k).
\end{equation}

\noindent
Here $x_k\in\mathbb R^{n}$, $u_k\in\mathbb R^{m}$, 
$p_k,q_k\in\mathbb R^{h}$ and $y_k\in\mathbb R^{l}$.  
Accordingly $A\in\mathbb R^{n\times n}$, 
$B_p\in\mathbb R^{n\times h}$, 
$B_u\in\mathbb R^{n\times m}$, 
$C_q\in\mathbb R^{h\times n}$, 
$D_{qp}\in\mathbb R^{h\times h}$, etc.; 
$\Gamma(\cdot)$ acts diagonally on $q_k$. The vector $q_k$ has components $q_k = [q_{k,1}, q_{k,2}, \ldots, q_{k,h}]^T$, and the diagonal nonlinearity $\Gamma(\cdot)$ is defined component-wise as $\Gamma(q_k) = [\phi_1(q_{k,1}), \phi_2(q_{k,2}), \ldots, \phi_h(q_{k,h})]^T$, where each $\phi_i : \mathbb{R} \to \mathbb{R}$ is a scalar nonlinear function.

Without the system output equation the SNOF can be formulated by
$\left[\begin{smallmatrix}x_{k+1} \\ q_k\end{smallmatrix}\right] = M \left[\begin{smallmatrix}x_k  \\ p_k\end{smallmatrix}\right]$,
with $p_k = \Gamma(q_k)$. The SNOF is well-posed if it fulfills the following well-posedness criterion:

\begin{definition}[Well‑posed SNOF {\cite{Kim2018StandardRA}}]\label{def:SNOF}
Let $\Gamma:\mathbb R^{h}\!\to\!\mathbb R^{h}$ be the memoryless non‑linearity in a SNOF and assume $\Gamma(0)=0$ and $\Gamma$ is injective. Then there exists a (diagonal) time‑varying gain  
$\Delta_k\in\mathbb R^{h\times h}$, $k\in\mathbb Z_{\ge 0}$, such that $\Gamma(q_k)=\Delta_k\,q_k
\quad\text{for every bounded }q_k\in\mathbb R^{h}.$ The SNOF is said to be \emph{well‑posed} if $R \;=\; I - M_{22}\,\Delta_k$
is invertible for all admissible $\Delta_k$.
\end{definition}

A popular method for proving the stability of systems like \eqref{eq:lure_system} and, by extension, \eqref{eq:SNOF}, is to construct a Lyapunov function. A Lyapunov function is a scalar function 
\[
V:\mathbb{R}^n\to\mathbb{R},
\]
that is used to certify the stability of an equilibrium point. For a discrete-time system \(x_{k+1}=f(x_k)\), the function \(V(x)\) must satisfy:
\begin{enumerate}
  \item \textbf{Positive Definiteness:} \(V(0)=0\) and \(V(x)>0\) for all \(x\neq0\).
  \item \textbf{Descent Condition:} The difference 
  \[
  \Delta V(x_k) = V(x_{k+1})-V(x_k)
  \]
  is negative definite, i.e., \(\Delta V(x_k) < 0\) for all \(x \neq 0\).
\end{enumerate}
Note that while quadratic Lyapunov functions (e.g., \(V(x)=x^\top P x\) with \(P>0\)) are common, they tend to be conservative when applied to systems with nonlinearities. Thus, composite Lyapunov functions that incorporate additional information such as sector bounds are often used to derive less conservative stability conditions. The idea of sector bounds is based on the work of \cite{1098316} and defined as:

\begin{definition}[Sector‑bounded non‑linearity {\cite{Gupta1994SomePA}}]\label{def:sb}
A time‑varying, memoryless non‑linearity $\psi(\sigma,k)$ is said to lie in the sector $[a,b]$ if $\bigl(\psi(\sigma,k)-a\sigma\bigr)\bigl(\psi(\sigma,k)-b\sigma\bigr)\le 0$ for every $\sigma\in\mathbb R$ and $k\in\mathbb Z_{\ge 0}$.
\end{definition}

Now consider the following Lyapunov candidate function utilizing sector-bounded nonlinearities, see \cite{kim2009robust}:
\begin{equation}\label{eq:modified_lyapunov}
\begin{aligned}
V(x_k) \;&=\; \tilde{x}_k^{\top} P \,\tilde{x}_k 
\;+\; 2 \sum_{i=1}^{h} Q_{ii}\int_{0}^{\,q_{k,i}}\phi_{i}(\sigma)\,d\sigma 
\;\\&\quad+ 2 \sum_{i=1}^{h} \tilde{Q}_{ii} \int_{0}^{\,q_{k,i}} \bigl[\xi\,\sigma \;-\; \phi_{i}(\sigma)\bigr]\,d\sigma,\\
\tilde{x}_k \;&\triangleq\; 
\left[\begin{smallmatrix*}
x_k\\
p_k\\
q_k
\end{smallmatrix*}\right], P^{\top} \;=\; P \;\triangleq\;
\left[\begin{smallmatrix*}
P_{11} & P_{12} & P_{13}\\
P_{12}^{\top} & P_{22} & P_{23}\\
P_{13}^{\top} & P_{23}^{\top} & P_{33}
\end{smallmatrix*}\right]\;\geq\; 0,
\end{aligned}
\end{equation} 
with \(P_{11} > 0\), and the diagonal elements of \(Q\) and \(\tilde{Q}\) satisfy \(Q_{ii} \;\ge\; 0, \tilde{Q}_{ii} \;\ge\; 0, \forall\,i=1,\ldots,h.\)
In \eqref{eq:modified_lyapunov}, the first term \(\tilde{x}_k^{\top} P\,\tilde{x}_k\) captures the energy of the augmented state, while the integrals contains information from the nonlinearities \(\phi_i(\sigma)\) and the sector-bound parameter \(\xi\), reducing conservatism in the stability analysis.

\cite{kim2009robust} showed, that using slope restriction bounds on the nonlinearity leads to a nonconservative stability theorem ensuring the difference of the candidate Lyapunov function of \eqref{eq:modified_lyapunov} is always negative. Slope restriction is hereby defined as:

\begin{definition}[Slope‑restricted non‑linearity {\cite{Nguyen2021RobustCT}}]\label{def:sr}
A time‑varying, memoryless non‑linearity $\psi(\sigma,k)$ is \emph{slope‑restricted} in the sector $[a,b]$ if $a \,\le\,
\displaystyle(\psi(\sigma,k)-\psi(\hat\sigma,k))/(\sigma-\hat\sigma)
\,\le\, b$ for all $\sigma\neq\hat\sigma\in\mathbb R$ and all $k\in\mathbb Z_{\ge 0}$.
\end{definition}

\begin{remark}
The sector-bounded and slope-restricted assumptions are fundamental requirements inherited from classical absolute stability theory \cite{kim2009robust}. These conditions are necessary mathematical prerequisites that enable the construction of the composite Lyapunov function in \eqref{eq:modified_lyapunov}. In practice, many common activation functions naturally satisfy these conditions, and these assumptions represent the current state-of-the-art for obtaining global stability guarantees.
\end{remark}

The sufficient stability criterion for a SNOF with slope-restricted and sector-bounded nonlinearities is stated by \cite{kim2009robust} as:

\begin{theorem}[Stability of a SNOF {\cite{kim2009robust}}]\label{th:stabtheorem}
Consider a SNOF whose memoryless non‑linearity
$\Gamma$ is sector bounded in $[0,\xi]$ and slope restricted in $[0,\mu]$,
and is continuous almost everywhere.  
The closed loop is \emph{globally asymptotically stable} if there exist a symmetric matrix $P=P^\top\succeq 0$ with a positive‑definite block $P_{11}\succ 0$, and  diagonal, positive‑semidefinite matrices $Q,\;\tilde Q,\;T,\;\tilde T,\;N\in\mathbb R^{h\times h}$ such that the LMI
\[
G \;=\;
\begin{bmatrix}
G_{11} & G_{12} & G_{13}\\
G_{12}^\top & G_{22} & G_{23}\\
G_{13}^\top & G_{23}^\top & G_{33}
\end{bmatrix}
\;<\; 0
\]
is satisfied.  
The explicit expressions for the blocks \(G_{ij}\) are stated in \cite{kim2009robust}.
\end{theorem}

The conditions stated in Theorem \ref{th:stabtheorem} correspond to a \emph{Linear Matrix Inequality (LMI)} formulation. Therefore, the stability conditions can be verified numerically by solving the LMI optimization problem. For solving one can utilize standard convex optimization libraries designed for LMIs. If the LMI solver successfully provides matrices $P,Q,\tilde{Q},T,\tilde{T},N$, the closed-loop system fulfills the conditions for global asymptotic stability.

\subsection{Neural Network Architectures}
Recurrent neural networks (RNNs) are widely used for modeling temporal dependencies in sequential data. In particular, \emph{gated} architectures, such as the Long Short-Term Memory (LSTM) network \cite{LSTM} and the Gated Recurrent Unit (GRU) \cite{GRU}, introduce internal gating mechanisms that mitigate the vanishing or exploding gradient problem. These gates dynamically regulate how information flows through the hidden states, allowing the network to capture long-range dependencies more effectively, which make them especially viable for virtual sensor usecases. 
We focus on the GRU due to its simpler architecture, which reduces analytical complexity while still using the essential gating dynamics. The analysis for LSTM architectures can be performed analogously, but involves additional gates and parameters, making the derivations more extensive. A vanilla GRU architecture is expressed as:

\begin{equation}\label{eq:GRU_original}
\begin{aligned}
r_k        &= \sigma\!\bigl(W_{r,x}x_k + b_{r,x} + W_{r,h}h_{k-1} + b_{r,h}\bigr),\\
z_k        &= \sigma\!\bigl(W_{z,x}x_k + b_{z,x} + W_{z,h}h_{k-1} + b_{z,h}\bigr),\\
\tilde h_k &= \tanh\!\bigl(W_{\tilde h,x}x_k + b_{\tilde h,x}
             + W_{\tilde h,rh}(r_k\odot h_{k-1}) + b_{\tilde h,rh}\bigr),\\
h_k        &= (1-z_k)\odot h_{k-1} + z_k\odot\tilde h_k .
\end{aligned}
\end{equation}

\noindent
$x_k\in\mathbb R^{n_x}$ is the input, 
$h_{k-1},h_k,r_k,z_k,\tilde h_k\in\mathbb R^{n_h}$ are hidden‑layer vectors.  
$W_{\!*}$ denotes weight matrices ($n_h\times n_x$ or $n_h\times n_h$), 
$b_{\!*}$ the bias vectors, 
$\sigma(\cdot)$ the logistic sigmoid, 
$\tanh(\cdot)$ the hyperbolic tangent, 
and $\odot$ the Hadamard (element‑wise) product, where $(A \odot B)_{ij} = A_{ij} \cdot B_{ij}$.

\subsection{Problem Statement}

The problem addressed in this paper is based of the setup discussed by \cite{10394670}. The challenge is to derive formal proofs of stability for a controlled discrete Multi-Input-Multi-Output (MIMO) plant, for which a machine learning-based virtual sensor estimates one or more of its output signals. The virtual sensor is in our case based on a gated recurrent network, e.g. a GRU. This setup is visualized in Fig. \ref{fig:problem}. In this figure, the discrete plant is denoted as $\Sigma^{(p)}_{k,\Delta}$, the machine‑learning–based virtual sensor by $\Sigma^{(n)}_{k}$, and the discrete controller by $\Sigma^{(c)}_{k}$.  The reference signal is $r_k$, the tracking error is $e_k$, and the plant input is $u^{(p)}_{k}$.  The overall plant output is $y^{(p)}_{k}$, which we partition into the measured part $y^{(p)}_{k,\Delta}$ and the virtual‑sensor estimate $y^{(n)}_{k}$, i.e.\ $y^{(p)}_{k}=\{\,y^{(p)}_{k,\Delta},\,y^{(n)}_{k}\}$.

\begin{figure*}[htb]
  \centering
  \includegraphics[width=2\columnwidth]{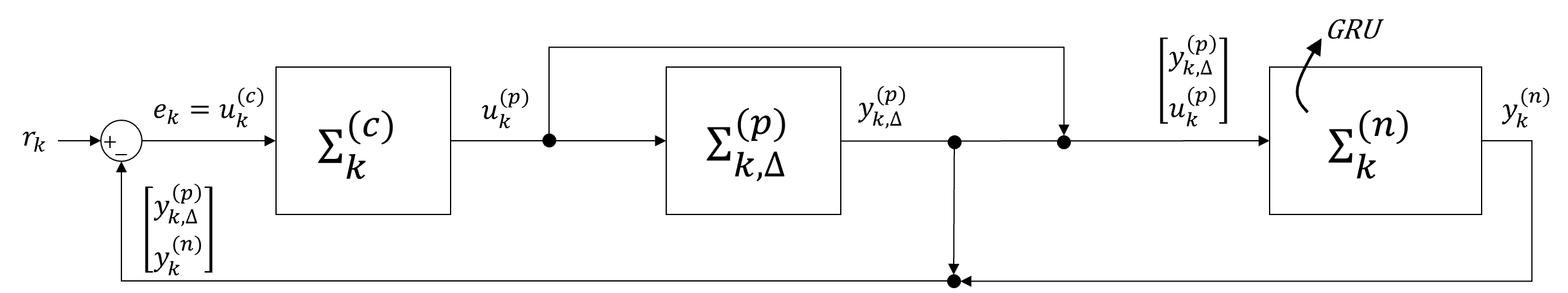}
  \caption{Illustration of the general problem of a discrete-time plant controlled by a discrete-time controller and a machine learning-based virtual sensor estimating a plant output signal \cite{10394670}.}
  \label{fig:problem}
\end{figure*}

Before we continue analyzing stability for such a system, we need to dig deeper into Hadamard gated recurrent neural networks and their incompatibility with the absolute stability theory and therefore the SNOF.

\section{Incompatibility of Hadamard Gated Networks with Absolute Stability Theory}\label{sec:hadamard}
Our goal is to adapt the stability theorem \ref{th:stabtheorem} for the problem in Fig. \ref{fig:problem}. The first step is to convert the Gated Recurrent Unit (GRU) into the underlying SNOF-like structure, expressed as:

\begin{equation}\label{eq:NNSNOF}
\begin{aligned}
\Sigma_k^{(n)}&:\;
\left\{
\begin{aligned}
x^{(n)}_{k+1} &= A^{(n)}x^{(n)}_k
                + B^{(n)}_p p^{(n)}_k
                + B^{(n)}_u u^{(n)}_k
                + \beta_x^{(n)},\\
q^{(n)}_k     &= C^{(n)}_q x^{(n)}_k
                + D^{(n)}_{qp} p^{(n)}_k
                + D^{(n)}_{qu} u^{(n)}_k
                + \beta_q^{(n)},\\
y^{(n)}_k     &= C^{(n)}_y x^{(n)}_k
                + D^{(n)}_{yp} p^{(n)}_k
                + D^{(n)}_{yu} u^{(n)}_k
                + \beta_o^{(n)},
\end{aligned}\right. \\
p^{(n)}_k &= \Gamma^{(n)}(q^{(n)}_k).
\end{aligned}
\end{equation}

\noindent
$x^{(n)}_k\in\mathbb R^{n}$ is the state, 
$u^{(n)}_k\in\mathbb R^{m}$ the input, 
$q^{(n)}_k,p^{(n)}_k\in\mathbb R^{h}$ input and output vectors of the nonlinear operator, 
and $y^{(n)}_k\in\mathbb R^{l}$ the network output.  
$\beta_{\!*}$ are bias vectors; 
all other matrices are sized conformably.

In the following we propose a SNOF-like formulation of the GRU and then outline the derivation steps that lead to it.

\begin{proposition}[GRU in SNOF‑like form]\label{pr:GRU2SNOF}
The standard GRU from \eqref{eq:GRU_original} admits the state–space representation

\begin{equation}\label{eq:GRU-SNOF-transformed}
\begin{aligned}
\begin{bmatrix}
h_{k+1}\\
q^{(h)}_{k}\\
q^{(\tilde h)}_{k}\\
q^{(rh)}_{k}\\
q^{(z)}_{k}\\
q^{(r)}_{k}
\end{bmatrix}
&=\left[
\begin{smallmatrix}
 I &  I & 0 & 0 & 0 & 0\\
-I &  0 & I & 0 & 0 & 0\\
 0 &  0 & 0 & W_{\tilde h,rh} & W_{\tilde h,x} & b_{\tilde h,x}+b_{\tilde h,rh}\\
 I &  0 & 0 & 0 & 0 & 0\\
W_{z,h} & 0 & 0 & 0 & W_{z,x} & b_{z,x}+b_{z,h}\\
W_{r,h} & 0 & 0 & 0 & W_{r,x} & b_{r,x}+b_{r,h}
\end{smallmatrix}\right]
\begin{bmatrix}
h_{k}\\
p^{(h)}_{k}\\
p^{(\tilde h)}_{k}\\
p^{(rh)}_{k}\\
x_{k}\\
1
\end{bmatrix},\\[4pt]
p^{(h)}_{k}       &= \sigma\!\bigl(q^{(z)}_{k}\bigr)\,\odot\,q^{(h)}_{k},\\
p^{(\tilde h)}_{k}&= \tanh\!\bigl(q^{(\tilde h)}_{k}\bigr),\\
p^{(rh)}_{k}      &= \sigma\!\bigl(q^{(r)}_{k}\bigr)\,\odot\,q^{(rh)}_{k}.
\end{aligned}
\end{equation}
\end{proposition}

\begin{proof}
Starting from the standard GRU formulation \eqref{eq:GRU_original}, we note that the state update combines linear transformations of the previous state and inputs with nonlinear activations. For clarity and without loss of generality (assuming a linear output layer), we decompose the hidden state update as
\begin{equation}
\begin{aligned}
    h_k &= h_{k-1} + z_k \odot (\tilde{h}_k-h_{k-1}), \\
    h_k &= h_{k-1} + p^{(h)}_k, \\
    p^{(h)}_k &= z_k \odot q^{(h)}_k.
\end{aligned}
\end{equation}
This expression is then recast into a SNOF-like notation:
\begin{equation}\label{eq:GRU_SNOF_1}
\begin{bmatrix}
h_{k+1}\\
q^{(h)}_{k}
\end{bmatrix}
=
\begin{bmatrix}
 I & I & 0\\
- I & 0 & I
\end{bmatrix}
\begin{bmatrix}
h_{k}\\
p^{(h)}_{k}\\
\tilde h_{k}
\end{bmatrix}.
\end{equation}

\noindent
The static relation is $p^{(h)}_{k}=z_{k}\,\odot\,q^{(h)}_{k}$. The candidate hidden state is defined by
\begin{equation}\label{eq:GRU_SNOF_2}
\begin{aligned}
 \tilde{h}_k &= p^{(\tilde{h})}_k = \Gamma(q_k^{(\tilde{h})}), \\
q_k^{(\tilde{h})} &= W_{\tilde{h},x}\, x_k + b_{\tilde{h},x} + W_{\tilde{h},rh}\,(r_k \odot h_{k-1}) + b_{\tilde{h},rh}, \\
&= W_{\tilde{h},x}\, x_k + b_{\tilde{h},x} + W_{\tilde{h},rh}\, p^{(rh)}_k + b_{\tilde{h},rh}, \\
p^{(rh)}_k &= r_k \odot q^{(rh)}_k.
\end{aligned}
\end{equation}

By merging \eqref{eq:GRU_SNOF_1} and \eqref{eq:GRU_SNOF_2}, we obtain an intermediary formulation that contains the state and the candidate hidden state:
\begin{equation}\label{eq:gru-snof-merge}
\begin{bmatrix}
h_{k+1}\\
q^{(h)}_{k}\\
q^{(\tilde h)}_{k}\\
q^{(rh)}_{k}
\end{bmatrix}
=\left[
\begin{smallmatrix}
 I &  I & 0 & 0 & 0 & 0\\
- I &  0 & I & 0 & 0 & 0\\
 0 &  0 & 0 & W_{\tilde h,rh} & W_{\tilde h,x} & b_{\tilde h,x}+b_{\tilde h,rh}\\
 I &  0 & 0 & 0 & 0 & 0
\end{smallmatrix}\right]
\begin{bmatrix}
h_{k}\\
p^{(h)}_{k}\\
p^{(\tilde h)}_{k}\\
p^{(rh)}_{k}\\
x_{k}\\
1
\end{bmatrix}.
\end{equation}

The next step is to include the gating signals \(z_k\) and \(r_k\). These are implemented as
\begin{equation}
\begin{aligned}
z_k &= p^{(z)}_k = \Gamma(q_k^{(z)}),\\q_k^{(z)} &= W_{z,x}\, x_k + b_{z,x} + W_{z,h}\, h_{k-1} + b_{z,h},\\
r_k &= p^{(r)}_k = \Gamma(q_k^{(r)}),\\q_k^{(r)} &= W_{r,x}\, x_k + b_{r,x} + W_{r,h}\, h_{k-1} + b_{r,h}.
\end{aligned}
\end{equation}
Substituting these relations back into the merged equations results in the final SNOF-like structure of \eqref{eq:GRU-SNOF-transformed}.

\end{proof}

While the overall structure can be cast in a SNOF-like form, the multivariate nonlinearities inherent in the GRU's gating mechanism conflict with the formulation of the scalar Lyapunov candidate function which underlies the theorem \ref{th:stabtheorem}.

\begin{proposition}[Hadamard Incompatibility with Lur\'e--Postnikov]
Let \(f:\mathbb{R}^n \to \mathbb{R}^n\) be defined by
\[
   f(x) \;=\; g(x)\;\odot\; h(x),
\]
where \(g\) and \(h\) are scalar functions applied componentwise to \(x\), each assumed to be sector-bounded and slope-restricted. In general, the mapping \(f\) is \emph{nonconservative}, so there does not exist a scalar storage (or energy) function \(U:\mathbb{R}^n\to\mathbb{R}\) with \(\nabla U=f\).
\end{proposition}

\begin{proof}
For a scalar function \(f:\mathbb{R}\to\mathbb{R}\), one can define a scalar potential
\[
   V(x) \;=\; \int_{0}^{x} f(s)\,ds
\]
uniquely because there is only a single coordinate of integration. However, in multiple dimensions a path-independent integral
\[
   U(x) \;=\; \int_{\gamma} f(s)\,\mathrm{d}s
\]
defines a well-defined global potential if and only if \(f\) is a gradient field, which requires that \(\nabla \times f = 0\) in \(\mathbb{R}^n\).  If \(f\) is constructed by the Hadamard (componentwise) product of two nontrivial vector functions \(g\) and \(h\), it typically has mixed partial derivatives that fail to agree, so its Jacobian is not symmetric.  This violates the condition for exactness (zero curl), making \(f\) nonconservative.

Because \(\int_{\gamma} f(s)\,\mathrm{d}s\) then depends on the path \(\gamma\), there is no single scalar function \(U\) whose gradient equals \(f\).  Hence, a Lur\'e--Postnikov argument relying on a unique storage function fails in general when using these Hadamard products.
\end{proof}

\begin{remark}
The incompatibility shown in the preceding proposition directly violates the requirements of Theorem~\ref{th:stabtheorem}. Specifically, the composite Lyapunov function in \eqref{eq:modified_lyapunov} relies on path-independent line integrals of the form $\int_0^{q_{k,i}} \phi_i(\sigma)\,d\sigma$ for each nonlinearity component. When $\phi$ involves Hadamard products as in $p_k^{(h)} = z_k \odot q_k^{(h)}$, the resulting multivariate mapping lacks a conservative vector field, making these integrals path-dependent. \textbf{Consequently, the LMI formulation in Theorem~\ref{th:stabtheorem} becomes infeasible:} the matrices $Q$ and $\tilde{Q}$ in the LMI cannot capture the energy contributions from path-dependent nonlinearities, as the integral terms in \eqref{eq:modified_lyapunov} are undefined. This mathematical impossibility forces any analysis of Hadamard-gated systems to resort to conservative outer bounds or abandon the Lyapunov approach entirely, resulting in either infeasible LMIs or overly restrictive stability conditions.
\end{remark}

\begin{remark}
In special or trivial cases, e.g. if each component of \(f\) depends only on one component of \(x\) such that the Jacobian is diagonal or block-diagonal, the field may remain conservative.  But for typical multi-coordinate gating - as in GRUs - cross-terms produce non-zero curl, violating the necessary condition for a global potential.
\end{remark} 

The inability to define a well-behaved closed-loop storage function means that the Hadamard product-based gating mechanism of the GRU (and similarly of the LSTM) cannot be transformed into a SNOF-compatible form suitable for the application of the stability theorem \ref{th:stabtheorem}. To address this incompatibility, we introduce the LP-GRNN — a modified gated recurrent architecture designed to preserve the benefits of GRUs while enabling transformation into the SNOF. This makes it compatible with nonconservative Lur’e-Postnikov stability analysis.

\section{Introduction of LP-GRNN}\label{sec:LP-GRNN}
As demonstrated in the previous section, the state-dependent multiplicative interaction of the standard update gate $z_k$ prevents the construction of a path-independent storage function required for absolute stability analysis. To overcome this structural incompatibility, we propose a modification that trades dynamic gating flexibility for strict structural compliance with the Luré-Postnikov framework.The core idea of the proposed \emph{Luré-Postnikov Gated Recurrent Neural Network} (LP-GRNN) is to decouple the gating coefficient from the current state signal. By replacing the time-varying, state-dependent update gate $z_k$ with a learnable but time-invariant mixing vector $\alpha$, we restore the affine nature of the state update equation. This modification effectively transforms the complex multiplicative nonlinearity into a simpler sector-bounded nonlinearity, making the architecture applicable to SNOF transformation while retaining the capability to learn optimal time constants for memory retention. The new LP-GRNN structure is defined as follows and visualized in Fig. \ref{fig:CGRU}.

\begin{equation}\label{eq:lp-grnn}
\begin{aligned}
r_k & = \sigma(W_{r,x} x_k + b_{r,x} + W_{r,h} h_{k-1} + b_{r,h}) ,\\
\tilde{h}_k &= \tanh(W_{\tilde{h},x} x_t + b_{\tilde{h},x} + G_{\tilde{h},r} r_k + b_{\tilde{h},r} \\ 
&\quad+ W_{\tilde{h},h} h_{k-1} + b_{\tilde{h},h}) ,\\
h_k &= [1-\sigma(\alpha)] \tilde{h}_k + \sigma(\alpha)  h_{k-1} .
\end{aligned}
\end{equation}

\begin{figure}
  \centering
  \includegraphics[width=1\columnwidth]{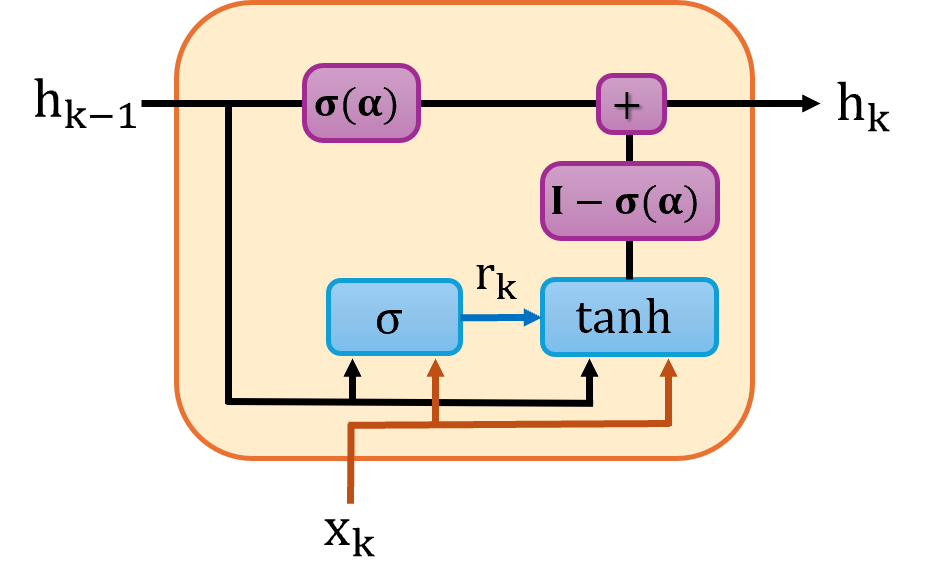}
  \caption{Architecture of the proposed LP-GRNN cell. The dynamic update gate $z_k$ of the standard GRU is replaced by a learnable time-invariant vector $\alpha$.}
  \label{fig:CGRU}
\end{figure}

It is important to acknowledge that fixing the update gate to a constant vector $\alpha$ simplifies the dynamic behavior compared to a standard GRU, where $z_k$ adapts at each time step based on the input. In the LP-GRNN, $\alpha$ acts as a static "leak rate" per neuron, optimized during training to capture the dominant temporal dynamics of the dataset. While this theoretically reduces the model's expressiveness regarding rapidly changing time constants, it is a necessary condition to obtain a well-posed SNOF without conservative over-approximation. As shown in Section \ref{sec:example}, this simplification does not lead to a statistically significant performance drop in practice. The compatibility with the SNOF-structure will be derived in the following.

\begin{proposition}[SNOF Representation of the LP-GRNN]

The LP-GRNN equation system of \eqref{eq:lp-grnn} can be equivalently rewritten as the SNOF:

\begin{equation}\label{eq:lp-grnn-snof}
\begin{aligned}    
\left[
\begin{smallmatrix}
x_k\\
q_k^{(\tilde h)}\\
q_k^{(r)}
\end{smallmatrix}\right]
&=\left[
\begin{smallmatrix}
I_h - \operatorname{dg}\sigma(\alpha) & \operatorname{dg}\sigma(\alpha) & 0_{h\times h} & 0_{h\times m} & 0_{h\times 1}\\
W_{\tilde h,h}            & 0_{h\times h}           & \tfrac12 I_h G_{\tilde h,r} & W_{\tilde h,x} & \tfrac12 G_{\tilde h,r} \mathbf1_h + b_{\tilde h}\\
\tfrac12 I_h W_{r,h}    & 0_{h\times h}           & 0_{h\times h}                 & \tfrac12 I_h W_{r,x} & \tfrac12\,I_h b_r
\end{smallmatrix}\right] \\
&\quad\quad\cdot\left[
\begin{smallmatrix}
x_{k+1}\\
p_k^{(\tilde h)}\\
p_k^{(r)}\\
u_k\\
1
\end{smallmatrix}\right],\\
b_{\tilde h} &= b_{\tilde h,x} + b_{\tilde h,r} + b_{\tilde h,h},
\qquad
b_r = b_{r,x} + b_{r,h}, \\
\tilde h_k &= p_k^{(\tilde h)} = \tanh\bigl(q_k^{(\tilde h)}\bigr),
\qquad
r_k        = p_k^{(r)}        = \tanh\bigl(q_k^{(r)}\bigr).
\end{aligned}
\end{equation}
\end{proposition}

\begin{proof}
We use the SNOF-like formulation procedure of the GRU resulting in \eqref{eq:GRU-SNOF-transformed} analogously to express the LP-GRNN in the following SNOF-like system $\mathrm{\Sigma}\mathit{^{(M)}_k}$:

\begin{equation}
\begin{aligned}
\mathrm{\Sigma}\mathit{^{(M)}_k} &:\left[\begin{smallmatrix}
\mathit{h_{k+1}} \\ 
\mathit{q_{i,k}} \\
\end{smallmatrix}\right]= \left[
\begin{smallmatrix}
A^{(M)} & B_p^{(M)} & B_u^{(M)} & \beta_x^{(M)}\\
C_q^{(M)} & D_{qp}^{(M)} & D_{qu}^{(M)} & \beta_q^{(M)}
\end{smallmatrix}
\right] \left[\begin{smallmatrix}
\mathit{h_k} \\
\mathit{p_{i,k}} \\
\mathit{x_k} \\
1
\end{smallmatrix}\right], i = 1,2,
\\ 
\left[\begin{smallmatrix}
\mathit{h_{k+1}} \\ 
\mathit{q^{(\tilde{h})}_k} \\
\mathit{q^{(r)}_k}
\end{smallmatrix}\right] &= 
\left[\begin{smallmatrix}
I - \operatorname{dg}\sigma(\alpha) & \operatorname{dg}\sigma(\alpha)\ & 0 & 0 & 0\\
W_{\tilde h,h} & 0 & G_{\tilde h,r} & W_{\tilde h,x} & b_{\tilde h,x}+b_{\tilde h,r}+b_{\tilde h,h}\\
W_{r,h}        & 0 & 0                 & W_{r,x}        & b_{r,x}+b_{r,h}
\end{smallmatrix}\right] 
\left[\begin{smallmatrix}
\mathit{h_k} \\
\mathit{p^{(\tilde{h})}_k} \\
\mathit{p^{(r)}_k} \\
\mathit{x_k} \\
1
\end{smallmatrix}\right] , \\
p^{(\tilde{h})}_k &= \tanh(q_k^{(\tilde{h})}) ,\\
p^{(r)} &= \sigma(q_k^{(r)}) .
\end{aligned}
\end{equation}

For correspondence with the sector boundedness and slope restriction criterion, we need to transform the sigmoid gating function of \( p_k^{(r)} \), as it is not sector bounded. Therefore we need to loop transform the system as stated in \cite{Kim2018StandardRA}. We use the fact that the hyperbolic tangent is slope restricted and sector bounded, and that $p_k^{(r)} = \sigma(q_k^{(r)}) = \frac{1}{2}\cdot\tanh(\frac{1}{2}q_k^{(r)})+\frac{1}{2}$. Thus, we can rewrite the sigmoid function as:

\begin{equation}
\begin{aligned}
    \begin{bmatrix}q^{'(r)}_k\\ p_k^{(r)}\end{bmatrix} &= \begin{bmatrix}
        0 & \frac{1}{2}I & 0\\
        \frac{1}{2}I & 0 & \frac{1}{2}I
        \end{bmatrix}\cdot\begin{bmatrix}p^{'(r)}_k\\q_k^{(r)}\\1\end{bmatrix} ,\\
    p^{'(r)}_k &= \Gamma(q^{'(r)}_k) = \tanh(q^{'(r)}_k) .
\end{aligned}
\end{equation}

We can rewrite this as a system in the form of:

\begin{equation}
\begin{aligned}
\mathrm{\Sigma}\mathit{^{(G)}_k} &:=\left\{ \begin{aligned} 
  q^{'(r)}_k &= \mathit{C^{(G)}_q}p^{'(r)}_k +  \mathit{D^{(G)}_{qp}}q_k^{(r)} + \mathit{D^{(G)}_{qu}}\cdot I \\ p_k^{(r)} &= \mathit{C^{(G)}_y}p^{'(r)}_k +  \mathit{D^{(G)}_{yp}}q_k^{(r)} + \mathit{D^{(G)}_{yu}}\cdot I
  \end{aligned} \right.,\\
  p^{'(r)}_k &= \tanh(q^{'(r)}_k) .
\end{aligned}
\end{equation}

The process of loop transformation is visualized in Fig. \ref{fig:loop_transformation}. Now consider a extended SNOF formulation given by:

\begin{equation}\label{eq:ex_SNOF}
\begin{aligned}
    \begin{bmatrix*}
        x_k \\ q_k
    \end{bmatrix*} &=     \begin{bmatrix*}
        M_{11} & M_{12} & M_{13} & M_{14} \\
        M_{21} & M_{22} & M_{23} & M_{24}
    \end{bmatrix*}    \begin{bmatrix*}
        x_{k+1} \\ p_k \\ u_k \\ 1
    \end{bmatrix*},\\ p_k &= \Gamma(q_k) = \tanh(q_k).
\end{aligned}
\end{equation}

We can describe our interconnected systems in this form, by using the Redheffer star product and absorbing system $\Sigma_k^{(G)}$ in $\Sigma_k^{(M)}$:

\begin{equation}\label{eq:NNSNOFpre}
\begin{aligned}
\begin{bmatrix}
h_k\\
q_{i,k}
\end{bmatrix}
&= M
\begin{bmatrix}
x_{k+1}\\
p_{i,k}\\
u_k\\
1
\end{bmatrix},\\
p^{(n)}_{i,k} &= \tanh\bigl(q^{(n)}_{i,k}\bigr).
\end{aligned}
\end{equation}

with the sub-blocks of \(M\) defined as:
\begin{equation}\label{eq:M}
\begin{aligned}
\left[\begin{smallmatrix*}M_{12} & M_{14} 
\end{smallmatrix*}\right]&=  B_p^{(G)}(I_{n_p\cdot h}-D_{yp}^{(G)}D_{qp}^{(M)})^{-1}\left[\begin{smallmatrix*}C_y^{(G)} & D_{yu}^{(G)} 
    \end{smallmatrix*}\right] \\ &\kern-2.5em+ \left[\begin{smallmatrix*}0&\beta_x^{(M)}\end{smallmatrix*}\right],\\
    \left[\begin{smallmatrix*}M_{11} & M_{13} 
\end{smallmatrix*}\right] &= \left[\begin{smallmatrix*}
        A^{(M)} & B_{u}^{(M)}
    \end{smallmatrix*}\right]\\&\kern-2.5em+B_{p}^{(M)}(I_{n_p\cdot h}-D_{yp}^{(G)}D_{qp}^{(M)})^{-1} D_{yp}^{(G)}\left[\begin{smallmatrix*}C_q^{(M)} &D_{qu}^{(M)}\end{smallmatrix*}\right], \\
    \left[\begin{smallmatrix*}M_{22} & M_{24} 
\end{smallmatrix*}\right] &=\left[\begin{smallmatrix*}
        C_q^{(G)} & D_{qu}^{(G)}
    \end{smallmatrix*}\right]\\&\kern-2.5em+D_{qp}^{(G)}(I_{n_p\cdot h}-D_{qp}^{(M)}D_{yp}^{(G)})^{-1}D_{qp}^{(M)}\left[\begin{smallmatrix*}C_y^{(G)} & D_{yu}^{(G)}+\beta_q^{(M)}
    \end{smallmatrix*}\right],\\
    \left[\begin{smallmatrix*}M_{21} & M_{23} 
\end{smallmatrix*}\right] &= D_{qp}^{(G)}(I_{n_p\cdot h}-D_{qp}^{(M)}D_{yp}^{(G)})^{-1}\left[\begin{smallmatrix*}C_q^{(M)} & D_{qu}^{(M)}
    \end{smallmatrix*}\right].
\end{aligned}
\end{equation}

\begin{figure}
  \centering
  \includegraphics[width=1\columnwidth]{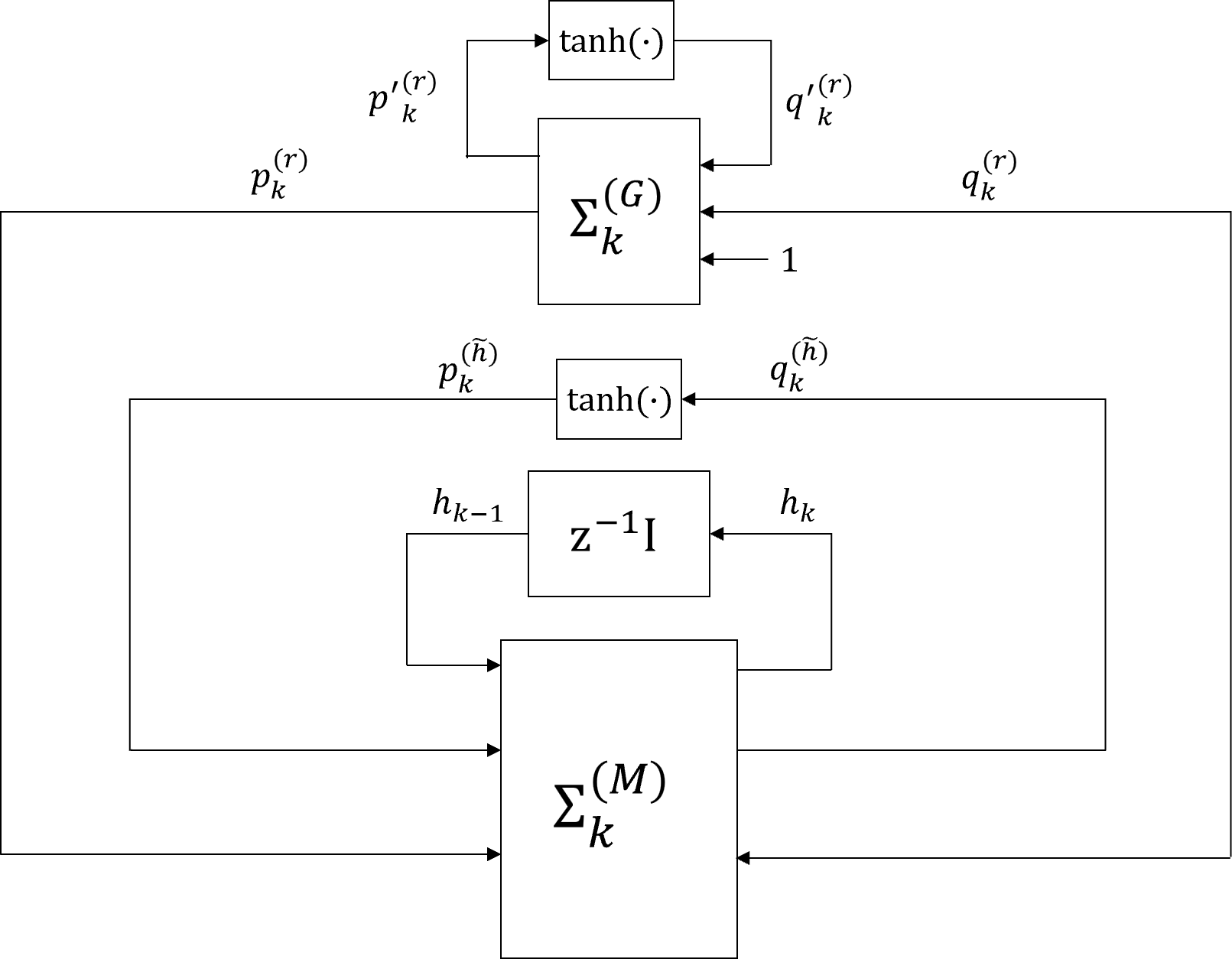}
  \caption{Loop transformation of the sigmoid activation. The identity 
  $\sigma(q) = \tfrac{1}{2}\tanh(\tfrac{1}{2}q) + \tfrac{1}{2}$ enables 
  reformulation of the gating mechanism using only $\tanh(\cdot)$ nonlinearities, which are sector bounded in $[0,1]$ and slope restricted, thereby satisfying the requirements of Definitions~\ref{def:sb} and~\ref{def:sr}.}
  \label{fig:loop_transformation}
\end{figure}

The calculation of the loop transformation leads to the proposed system representation of the LP-GRNN. \end{proof}

This transformed system now includes only hyperbolic tangent nonlinearities, which are sector-bounded and slope-restricted within [0,1] and satisfy the Definition \ref{def:sb} and Definition \ref{def:sr}, as shown in \cite{10394670}. As the nonlinearities and the structure now aligns with the SNOF framework, we show the well-posedness of the in \eqref{eq:lp-grnn-snof} proposed system:

\begin{proposition}
    The system proposed in \eqref{eq:lp-grnn-snof} is a well-posed SNOF, as defined by Definition \ref{def:SNOF}.
\end{proposition}
\begin{proof}
For the LP‑GRNN all static nonlinearities are component‑wise
\(\phi_i(\,\cdot\,)=\tanh(\,\cdot\,)\).
Because $\tanh$ is strictly monotone, odd and satisfies $\tanh(0)=0$,
it is injective and lies in the sector $[0,1]$ with slope restriction
$0\le\phi'_i(\sigma)\le 1$ for every $\sigma\in\mathbb{R}$.
Define
\[
\Delta_{k,ii}\;=\;
\begin{cases}
\dfrac{\phi_i(q_{k,i})}{q_{k,i}}, & q_{k,i}\neq 0,\\[6pt]
\phi'_i(0)=1, & q_{k,i}=0,
\end{cases}
\qquad i=1,\dots,n_q .
\]
Then $\Delta_k=\operatorname{dg}(\Delta_{k,11},\dots,\Delta_{k,n_qn_q})$
satisfies $\Gamma(q_k)=\Delta_k q_k$ for every $k\in\mathbb{Z}_{\ge 0}$,
and by construction $0\le\Delta_{k,ii}\le 1$. From \eqref{eq:lp-grnn-snof} and \eqref{eq:ex_SNOF} the $M_{22}$ block of the SNOF is
\[
M_{22}=\begin{bmatrix}
\mathbf 0_{h\times h} & \tfrac12\,G_{\tilde h,r}\\[2pt]
\mathbf 0_{h\times h} & \mathbf 0_{h\times h}
\end{bmatrix},
\quad
G_{\tilde h,r}\in\mathbb{R}^{h\times h}.
\]
Hence \(M_{22}\) is \emph{nilpotent of index 2} because \(M_{22}^2=0\). For any diagonal \(\Delta_k\) one has
\[
N\;:=\;M_{22}\Delta_k, \qquad N^2
    =M_{22}\Delta_k M_{22}\Delta_k
    =M_{22}^2\Delta_k^2=0.
\]
Because $N^2 = 0$ (nilpotent of index 2), we have
\[
(I-N)(I+N) \;=\; I - N^2 \;=\; I
\] and \[
(I+N)(I-N) \;=\; I - N^2 \;=\; I,
\]
hence $R^{-1} = I+N$. So $\det R = 1$ and $R$ is invertible for \emph{every} admissible $\Delta_k$. Because a diagonal \(\Delta_k\) satisfying $\Gamma(q_k)=\Delta_k q_k$ exists for all $k$, and $R=I-M_{22}\Delta_k$ is always nonsingular, Definition \ref{def:SNOF} is fulfilled. Therefore, the LP‑GRNN in \eqref{eq:lp-grnn-snof} is a well‑posed SNOF.
\end{proof}

We successfully proposed a gated recurrent architecture compatible with the absolute stability theory. We can therefore proceed with the stability analysis of the problem stated in \ref{sec:pre}. 

\section{Closed-Loop Stability Analysis}\label{sec:stability}

Having demonstrated the incompatibility of standard Hadamard gating with Luré-Postnikov theory and introduced the LP-GRNN as a compatible alternative, we now develop a comprehensive stability analysis framework for the closed-loop system in Fig.~\ref{fig:problem}. Our approach uses the stability theorem of \cite{kim2009robust} (Theorem~\ref{th:stabtheorem}) but requires several novel theoretical developments to bridge modern virtual sensor architectures with classical absolute stability theory. We do not modify the existing stability theorem itself, but make gated recurrent architectures compatible while preserving their modeling capabilities. Specifically, we derive: (i) a holistic SNOF transformation that unifies plant, controller, and LP-GRNN dynamics into a single analyzable system, (ii) well-posedness conditions for this integrated system, and (iii) verification that the resulting closed-loop satisfies the structural requirements for applying the LMI-based stability conditions. 

\begin{proposition}[Holistic SNOF of the problem setup]
The closed loop system presented in Fig. \ref{fig:problem} and stated in Section \ref{sec:pre} can be formulated as SNOF like:
\begin{equation}\label{eq:CPNNSNOF}
 \begin{aligned} 
  \left[\begin{smallmatrix*}
  x^{(c)}_{k+1} \\
  x^{(p)}_{k+1} \\
  x^{(n)}_{k+1} 
  \end{smallmatrix*}\right] &= 
  \left[\begin{smallmatrix*}
  A^{(c)} & 0& 0\\
  0& A^{(p)} & 0\\
  0& B^{(n)}_{u}\,\sigma_{i}\,\tilde C^{(p)}_{\Delta} & A^{(n)} 
  \end{smallmatrix*}\right]
  \left[\begin{smallmatrix*}
  x^{(c)}_{k} \\
  x^{(p)}_{k} \\
  x^{(n)}_{k} 
  \end{smallmatrix*}\right]\\
  &\quad+
  \left[\begin{smallmatrix*}
  B^{(c)}_{p} & 0\\
  B^{(p)} & 0\\
  B^{(n)}_{u}\,\sigma_{i}\,\tilde D_{\Delta}^{(p)} & B^{(n)}_{p} 
  \end{smallmatrix*}\right]
  \left[\begin{smallmatrix*}
  p^{(c)}_{k} \\
  p^{(n)}_{k}
  \end{smallmatrix*}\right] 
  + 
  \left[\begin{smallmatrix*}
  B^{(c)}_{u} \\
  0\\ 
  0
  \end{smallmatrix*}\right]
  u^{(c)}_{k} \\&\quad+ 
  \left[\begin{smallmatrix*}
  0\\
  0\\
  \beta_{x} + B^{(n)}_{u}\,\theta_{i}
  \end{smallmatrix*}\right],\\[1em]
  \left[\begin{smallmatrix*}
  q^{(c)}_{k} \\
  q^{(n)}_{k}
  \end{smallmatrix*}\right] &= 
  \left[\begin{smallmatrix*}
  C^{(c)}_{q} & 0 & 0 \\
  0 & D^{(n)}_{qu}\,\sigma_{i}\,\tilde C^{(p)}_{\Delta} & C^{(n)}_{q}
  \end{smallmatrix*}\right]  
  \left[\begin{smallmatrix*}
  x^{(c)}_{k} \\
  x^{(p)}_{k} \\
  x^{(n)}_{k} 
  \end{smallmatrix*}\right]\\
  &\quad+
  \left[\begin{smallmatrix*}
  D^{(c)}_{qp} & 0\\
  D^{(n)}_{qu}\,\sigma_{i}\,\tilde D_{\Delta}^{(p)} & D^{(n)}_{qp}
  \end{smallmatrix*}\right]
  \left[\begin{smallmatrix*}
  p^{(c)}_{k} \\
  p^{(n)}_{k}
  \end{smallmatrix*}\right]
  + 
  \left[\begin{smallmatrix*}
  D^{(c)}_{qu} \\
  0
  \end{smallmatrix*}\right]
  u^{(c)}_{k} \\&\quad+
  \left[\begin{smallmatrix*}
  0 \\
  \beta_{h} + D^{(n)}_{qu}\,\theta_{i}
  \end{smallmatrix*}\right],\\[1em]
  \left[\begin{smallmatrix*}
  y^{(p)}_{k,\Delta}\\
  y^{(n)}_{k}
  \end{smallmatrix*}\right] &= 
  \left[\begin{smallmatrix*}
  0 & C^{(p)}_{\Delta} & 0 \\
  0 & \sigma_{o}\,D^{(n)}_{yu}\,\sigma_{i}\,\tilde C^{(p)}_{\Delta} & \sigma_{o}\,C^{(n)}_{y}
  \end{smallmatrix*}\right]
  \left[\begin{smallmatrix*}
  x^{(c)}_{k} \\
  x^{(p)}_{k} \\
  x^{(n)}_{k} 
  \end{smallmatrix*}\right]\\
  &\quad+
  \left[\begin{smallmatrix*}
  D_{\Delta}^{(p)} & 0 \\
  \sigma_{o}\,D^{(n)}_{yu}\,\sigma_{i}\,\tilde D_{\Delta}^{(p)} & \sigma_{o}\,D^{(n)}_{yp}
  \end{smallmatrix*}\right]
  \left[\begin{smallmatrix*}
  p^{(c)}_{k} \\
  p^{(n)}_{k}
  \end{smallmatrix*}\right]
  + 
  \left[\begin{smallmatrix*}
  0 \\
  0
  \end{smallmatrix*}\right]
  u^{(c)}_{k} \\&\quad+
  \left[\begin{smallmatrix*}
  0 \\
  \sigma_{o}\,\beta_{o} + \sigma_{o}\,D^{(n)}_{yu}\,\theta_{i} + \theta_{o}
  \end{smallmatrix*}\right],\\[1em]
  p^{(n)}_{k} &= \Gamma^{(n)}(q^{(n)}_{k}), 
  \quad
  p^{(c)}_{k} = \Gamma^{(c)}(q^{(c)}_{k}).
\end{aligned}
\end{equation}
\end{proposition}

\begin{proof}
Analogous to the transformation discussed by \cite{10394670}, we can convert our plant and controller into the SNOF.
Based on the problem setup in Fig. \ref{fig:problem}, we consider the plant $\mathrm{\Sigma}\mathit{^{(p)}_{k,\mathrm{\Delta}}}$ given by
\begin{equation}\label{eq:SSPLANT}
\begin{aligned}
\Sigma^{(p)}_{k,\Delta}:\;
\left\{
\begin{aligned}
x^{(p)}_{k+1} &= A^{(p)}\,x^{(p)}_{k} + B^{(p)}\,u^{(p)}_{k},\\
\widetilde y^{(p)}_{k,\Delta}
  &= \widetilde C^{(p)}_{\Delta}\,x^{(p)}_{k}
   + \widetilde D^{(p)}_{\Delta}\,u^{(p)}_{k},
\end{aligned}
\right. \\
\widetilde C^{(p)}_{\Delta}
= \begin{bmatrix}
C^{(p)}_{\Delta}\\[2pt]
0_{\,s\times t}
\end{bmatrix},
\quad
\widetilde D^{(p)}_{\Delta}
= \begin{bmatrix}
D^{(p)}_{\Delta}\\[2pt]
I_{s}
\end{bmatrix}.
\end{aligned}
\end{equation}
\noindent
Here $x^{(p)}_{k}\in\mathbb R^{t}$, 
$u^{(p)}_{k}\in\mathbb R^{s}$, 
$y^{(p)}_{k,\Delta}\in\mathbb R^{\,v-1}$, 
and $A^{(p)}\in\mathbb R^{t\times t}$, 
$B^{(p)}\in\mathbb R^{t\times s}$, 
$C^{(p)}_{\Delta}\in\mathbb R^{(v-1)\times t}$, 
$D^{(p)}_{\Delta}\in\mathbb R^{(v-1)\times s}$.  
The plant dimension parameters $t,s,v$ are model‑specific, with $m=v-1+s$.  

Scaling of both plant and network outputs gives
\begin{equation}\label{eq:P-NN-relation}
\begin{aligned}
\widetilde{\widetilde y}^{(p)}_{k,\Delta}
&= \sigma_i\bigl(\widetilde C^{(p)}_{\Delta}\,x^{(p)}_{k}
   + \widetilde D^{(p)}_{\Delta}\,u^{(p)}_{k}\bigr)
  + \theta_i,\\
y^{(n)}_{k}
&= \sigma_o\bigl(C^{(n)}_y\,x^{(n)}_{k}
   + D^{(n)}_{yp}\,p^{(n)}_{k}
   + D^{(n)}_{yu}\,u^{(n)}_{k}
   + \beta_o\bigr)
  + \theta_o.
\end{aligned}
\end{equation}

\noindent
The scaling parameters $\sigma_i,\theta_i,\sigma_o,\theta_o$ are computed from the training data (e.g. via standard or min–max scaling).  
We can also transform our discrete controller into the SNOF. As in \cite{10394670}, we assume the controller output to be constraint by min-max values. Therefore, the controller can be represented as follows:
\begin{equation}\label{eq:CSNOF}
\begin{aligned}
\Sigma_k^{(c)}&:\;
\left\{
\begin{aligned}
x_{k+1}^{(c)} &= A^{(c)}\,x_k^{(c)} + B_p^{(c)}\,p_k^{(c)} + B_u^{(c)}\,u_k^{(c)},\\
q_k^{(c)}     &= C_q^{(c)}\,x_k^{(c)} + D_{qp}^{(c)}\,p_k^{(c)} + D_{qu}^{(c)}\,u_k^{(c)}, 
\end{aligned}\right.\\
p_k^{(c)} &= \Gamma^{(c)}\!\bigl(q_k^{(c)}\bigr) = y_k^{(c)},\\
\Gamma^{(c)}\!\bigl(q_k^{(c)}\bigr)
&= \min\!\bigl\{\max\!\bigl(y_{\min}^{(c)},\,q_k^{(c)}\bigr),\,y_{\max}^{(c)}\bigr\}.
\end{aligned}
\end{equation}

\noindent
Here $x_k^{(c)},\,q_k^{(c)},\,p_k^{(c)},\,y_k^{(c)}\in\mathbb R^{s}$ and 
$u_k^{(c)}\in\mathbb R^{v}$.  
The matrices satisfy 
$A^{(c)},B_p^{(c)},C_q^{(c)},D_{qp}^{(c)}\in\mathbb R^{s\times s}$,
$B_u^{(c)},D_{qu}^{(c)}\in\mathbb R^{s\times v}$.  
The nonlinearity $\Gamma^{(c)}$ enforces the controller output limits 
$y_{\min}^{(c)}$ and $y_{\max}^{(c)}\in\mathbb R^{s}\,$.

As shown in \cite{10394670}, the nonlinearities of the controller satisfy the Definition \ref{def:sb} and Definition \ref{def:sr}. For the virtual sensor system $\Sigma^{(n)}_{k}$, we use the LP-GRNN in SNOF form as outlined in \eqref{eq:lp-grnn-snof}, which can be stated as \eqref{eq:NNSNOF}. We can combine the systems $\Sigma^{(p)}_{k,\Delta}$, $\Sigma^{(c)}_{k}$, and $\Sigma^{(n)}_{k}$. Utilizing the problem specific identities $p^{(c)}_{k} = u^{(p)}_k$ and $\Tilde{\Tilde{y}}^{(p)}_{k,\mathrm{\Delta}}=u^{(n)}_k$ leads to the proposed formulation of the holistic SNOF.
\end{proof}

We now have a closed-loop SNOF representation of our problem. In the following, we will show that this system is a well-posed SNOF.

\begin{proposition}
    The system as outlined in \eqref{eq:CPNNSNOF} is a well-posed SNOF as defined by Definition \ref{def:SNOF}.
\end{proposition}

\begin{proof}
Write the holistic SNOF of \eqref{eq:CPNNSNOF} as
\[
\begin{bmatrix}x_{k+1}\\ q_k\end{bmatrix}
=
M\!
\begin{bmatrix}x_k\\ p_k\end{bmatrix},
\qquad
p_k=\Gamma(q_k),
\]
and partition \(p_k=[p^{(c)\!\top}_k,\,p^{(n)\!\top}_k]^\top\),
\(q_k=[q^{(c)\!\top}_k,\,q^{(n)\!\top}_k]^\top\).
The sub-matrix that couples the nonlinearities to themselves is

\[
M_{22}\;=\;
\begin{bmatrix}
 \underbrace{D^{(c)}_{qp}}_{=\,0} & 0 \\
 \boxed{H}=D^{(n)}_{qu}\,\sigma_i\tilde D^{(p)}_{\!\Delta}
          & \boxed{J}=D^{(n)}_{qp}
\end{bmatrix},
\]
Here \(D^{(c)}_{qp}=0\), because we chose a controller that
first sums (\(C_q^{(c)}x+B_u^{(c)}u\)) and then applies the output saturation; practically that means there is no algebraic loop from the saturation output straight back to its own input. Moreover, from \eqref{eq:lp-grnn-snof}, \(J\) has the explicit off-diagonal form
\[
J=\begin{bmatrix}0 & \tfrac12\,G_{\tilde h,r}\\[2pt]0 & 0\end{bmatrix},
\,\text{so }\;
\operatorname{dg}J=0.
\]
Let \(\Delta_k=\operatorname{dg}(\Delta_c,\Delta_n)\in[0,1]^{\,\times}\)
be the diagonal slope matrix (Definition \ref{def:SNOF}). Because \(J\) has a zero diagonal,  \(\operatorname{dg}(J\Delta_n)=0\). Therefore
\[
N=M_{22}\Delta_k=
\begin{bmatrix}
0 & 0\\
H\Delta_c & J\Delta_n
\end{bmatrix}
\]
is strictly block-lower-triangular and its main diagonal is
the zero matrix. The matrix \(R=I-N\) is therefore also block-lower-triangular with identity blocks on its main diagonal. Hence \(\det R=1\) and \(R\) is nonsingular for \emph{every} admissible \(\Delta_k\). Because a diagonal \(\Delta_k\) satisfying \(\Gamma(q_k)=\Delta_k q_k\) exists for all \(k\) (\(\tanh\) and the min-max-saturation are both injective and sector/slope-restricted)
and \(R\) is always invertible, Definition \ref{def:SNOF} is fulfilled: the closed-loop SNOF is well-posed.
\end{proof}

\begin{remark}
If, for a specific application, the controller has to include a (non-zero) diagonal feed-through \(D^{(c)}_{qp}\), we can
replace the assumption \(D^{(c)}_{qp}=0\) by the standard small-gain condition \(\|D^{(c)}_{qp}\|_\infty<1\) and the same block-triangular argument holds without change.
\end{remark}

We have therefore successfully transformed the closed-loop problem setup into a testable, well-posed SNOF structure including a gated recurrent network as virtual sensor. We can now apply Theorem~\ref{th:stabtheorem} to check for global asymptotic stability.  

\section{Evaluation}\label{sec:example}
In this section, we evaluate the above derived stability analysis. To this end, we first discuss, how the introduced LP-GRNN architectures compares on a benchmark time-series regression task to vanilla GRU and LSTM architectures. Afterwards, we provide an illustrative numerical example for the stability proof of a controlled linearized boiler plant, for which one output signal is estimated by a LP-GRNN based virtual sensor.

\subsection{LP-GRNN Benchmark}
We evaluate the performance of our proposed LP-GRNN neural network in a benchmark setting. Therefore, we compare our model structure against state-of-the-art methods using the NASA Turbofan Jet Engine dataset. The NASA Turbofan Jet Engine dataset \cite{pjh5-p424-23} is widely used for time-series prediction tasks. The dataset contains run-to-failure data from multiple engines. For each engine operational settings and sensor measurements are recorded over time. The target is to estimate the Remaining Useful Life (RUL) of an engine by utilizing its available data. We remark, that the purpose of the evaluation is a comparison of RNN architectures on a regression problem, not a state of the art performance on RUL problems.

For our experiments, we do a basic preprocessing by excluding redundant sensors and by normalizing the features to an interval of [-1, 1]. We cap the start value of the RUL to 120 and create a time series data set with a window length of 30 time steps. By utilizing the system failure time per engine, we can deduct the RUL for each time step. The RUL itself is also min-max scaled to an interval of [0,1]. We compare our proposed LP-GRNN model architecture with the following baseline architectures:
\begin{itemize}
    \item Multi-Layer Perceptron (MLP)
    \item Long Short-Term Memory (LSTM)
    \item Gated-Recurrent-Unit (GRU)
\end{itemize}

We utilize an NVIDIA GTX 4070 GPU for all experiments. All models are implemented using PyTorch. Their hyperparameters are optimized by utilizing the "optuna" library. We did 50 trials of 100 epochs of training with hyperparameter ranges for the recurrent architectures, i.e. a hidden Size range of $2^3$ to $2^{10}$ and a learning rate range of $1e^{-4}$ to $1e^{-2}$. For the MLP we decided to allow higher hidden sizes, to enable it to have a similar number of parameters as the other architectures resulting a hidden Size range of $2^8$ to $2^{15}$ and a learning rate range of $1e^{-4}$ to $1e^{-2}$.

We used the pytorch integrated learning rate scheduler "ReduceLROnPlateau" and the root-mean-squared-error (RMSE) loss for training. We note, that usually, for RUL prediction an overestimation is not desired and heavily penalized using a z-score metric, as proposed in the reference paper of the dataset \cite{4711414}. Contrary, for the tasks of virtual sensors, the MSE or RMSE is the main loss function to consider as we do not care about over- or underpredicting a signal. We will therefore solely analyze the RMSE performance of the models, as this is the relevant metric for our case.

To ensure robust performance estimates under stochastic training conditions, we conducted 50 independent trials for each architecture. An automated outlier detection using a consensus of IQR, Z-Score, and MAD methods was applied to filter anomalous training trajectories before aggregation. Table \ref{tab:benchmark_results} summarizes the predictive accuracy on the test set. The proposed LP-GRNN achieves a mean Root Mean Squared Error (RMSE) of 20.69 with a standard deviation of 2.99. While this average error is marginally lower than that of the vanilla GRU (21.71), the performance difference was confirmed to be \textit{statistically insignificant} (Wilcoxon signed-rank test, $p = 0.24 > 0.05$). These results empirically validate the design choice made in Section \ref{sec:LP-GRNN}. Despite replacing the dynamic Hadamard gating with the static affine update vector $\alpha$, the LP-GRNN maintains a predictive accuracy on par with the vanilla GRU. This suggests that for many practical control and sensing tasks, the learned time-constants provided by $\alpha$ are sufficient to capture the system dynamics, making the LP-GRNN a viable, stability-certifiable alternative to standard gated RNNs.

\begin{table}[htb]
    \centering
    \caption{Performance Comparison (50 Trials, Outliers Removed)}
    \label{tab:benchmark_results}
    \begin{tabular}{lcc}
        \hline
        \textbf{Method} & \textbf{Mean RMSE} & \textbf{Std. Dev.} \\
        \hline
        MLP     & 26.77 & 2.66 \\
        LSTM    & 23.21 & 3.09 \\
        GRU     & 21.71 & 3.79 \\
        \textbf{LP-GRNN} & \textbf{20.69} & \textbf{2.99} \\
        \hline
    \end{tabular}
\end{table}

The performance distribution across all completed trials is visualized in Fig. \ref{fig:performance_comparison}. Notably, the LP-GRNN exhibits a competitive median performance and a compact interquartile range, confirming its robustness on par with state-of-the-art recurrent models. 

\begin{figure}[!t]
\centering
\includegraphics[width=\columnwidth]{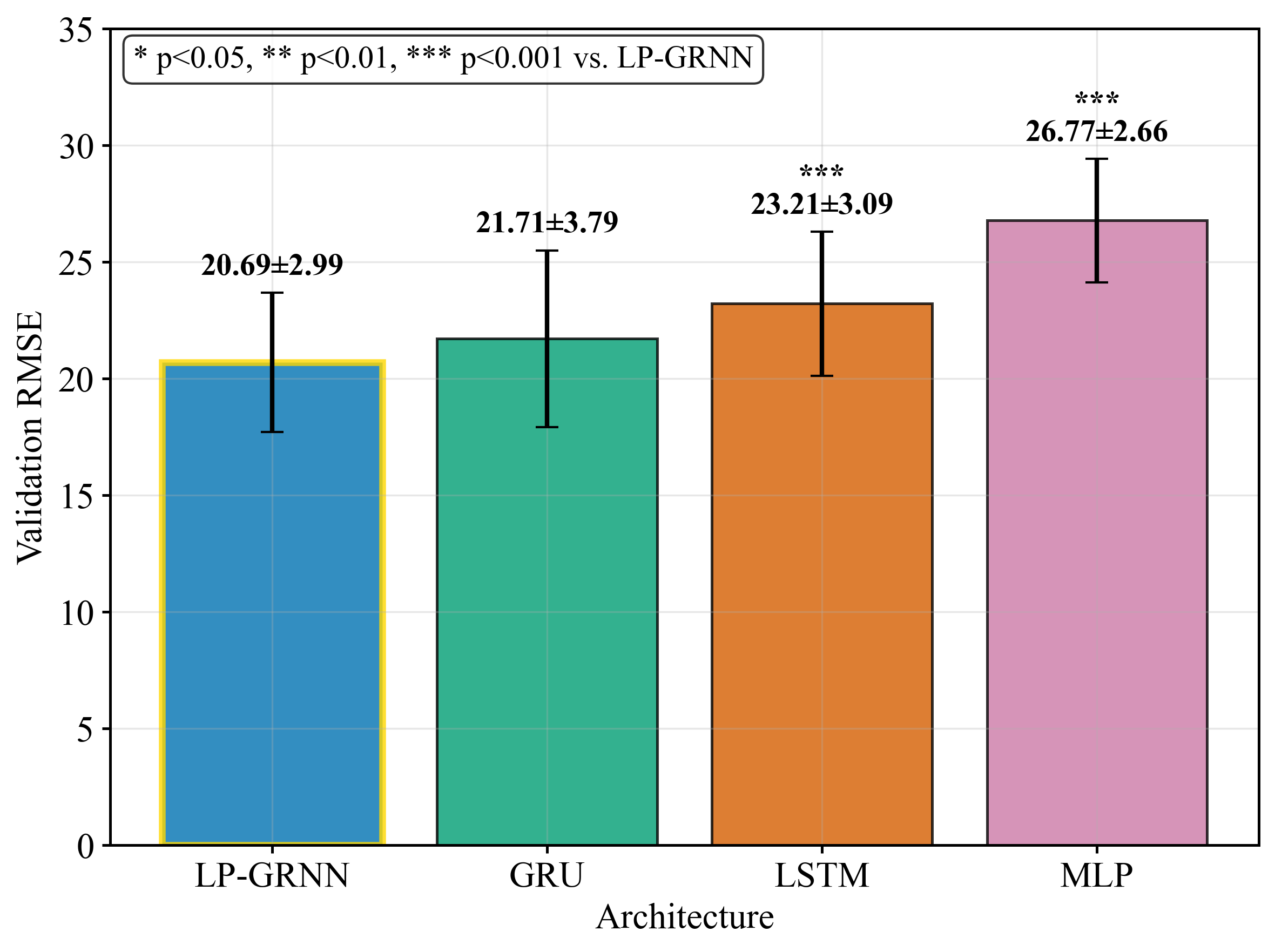}
\caption{Performance comparison of RNN architectures on the CMAPSS dataset. Bars represent mean RMSE with error bars showing the interquartile range (IQR).}
\label{fig:performance_comparison}
\end{figure}

\begin{figure}[!t]
\centering
\includegraphics[width=\columnwidth]{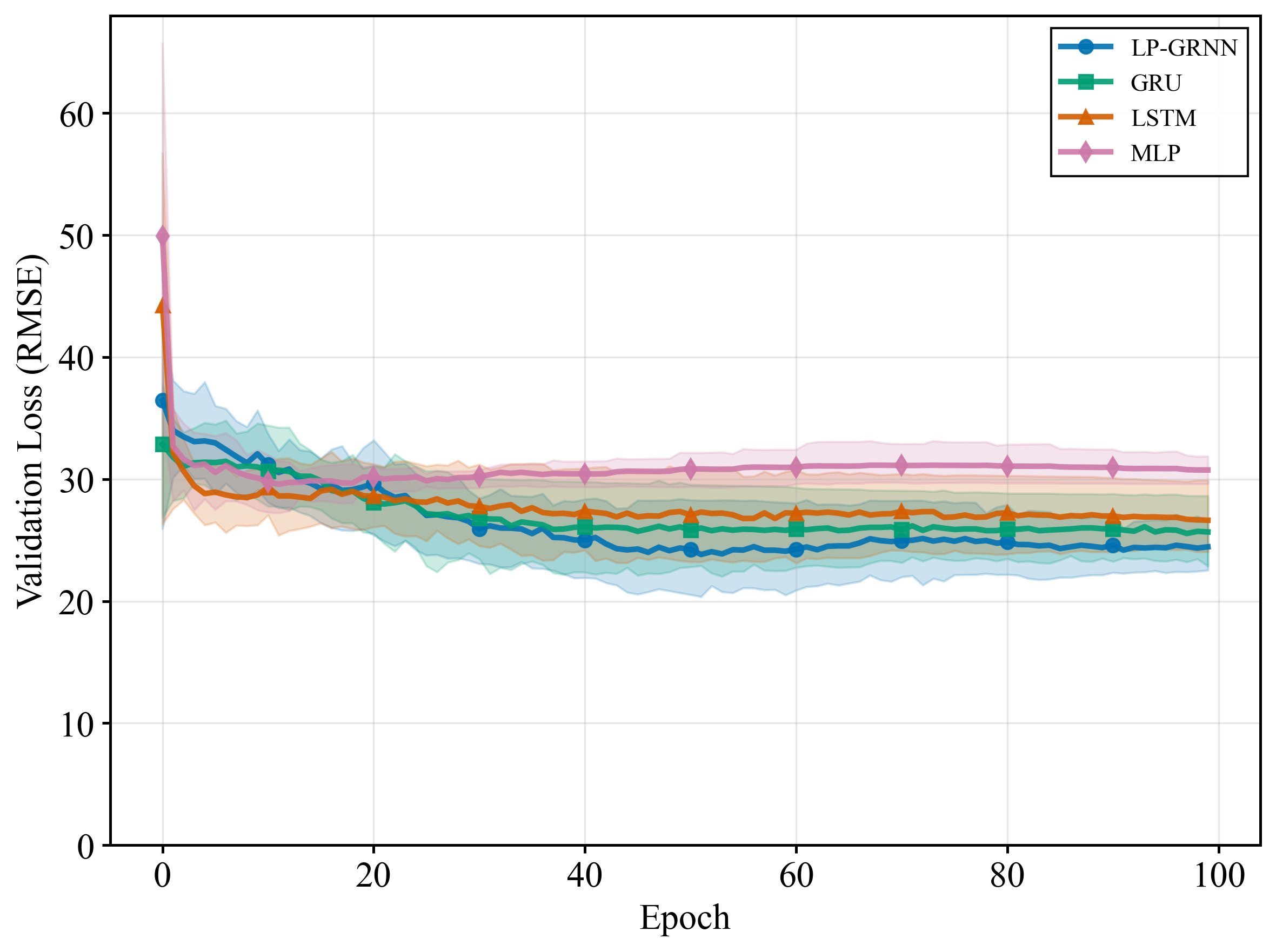}
\caption{Training convergence analysis showing validation loss across epochs. Solid lines represent mean validation loss with shaded regions indicating the interquartile range (IQR).}
\label{fig:convergence_analysis}
\end{figure}

We also analyzed the dynamic properties. Fig. \ref{fig:gradient_norm} illustrates the evolution of the gradient norms during training. As detailed in Table \ref{tab:dynamics_metrics}, the LP-GRNN maintains the lowest mean gradient norm of 24.45, representing a reduction of approximately 27\% compared to the standard GRU (33.46) and 44\% compared to the MLP (43.37). This significant reduction in gradient magnitude implies a more stable optimization trajectory, rendering the proposed architecture less prone to sudden, destabilizing updates. 

\begin{figure}[!t]
\centering
\includegraphics[width=\columnwidth]{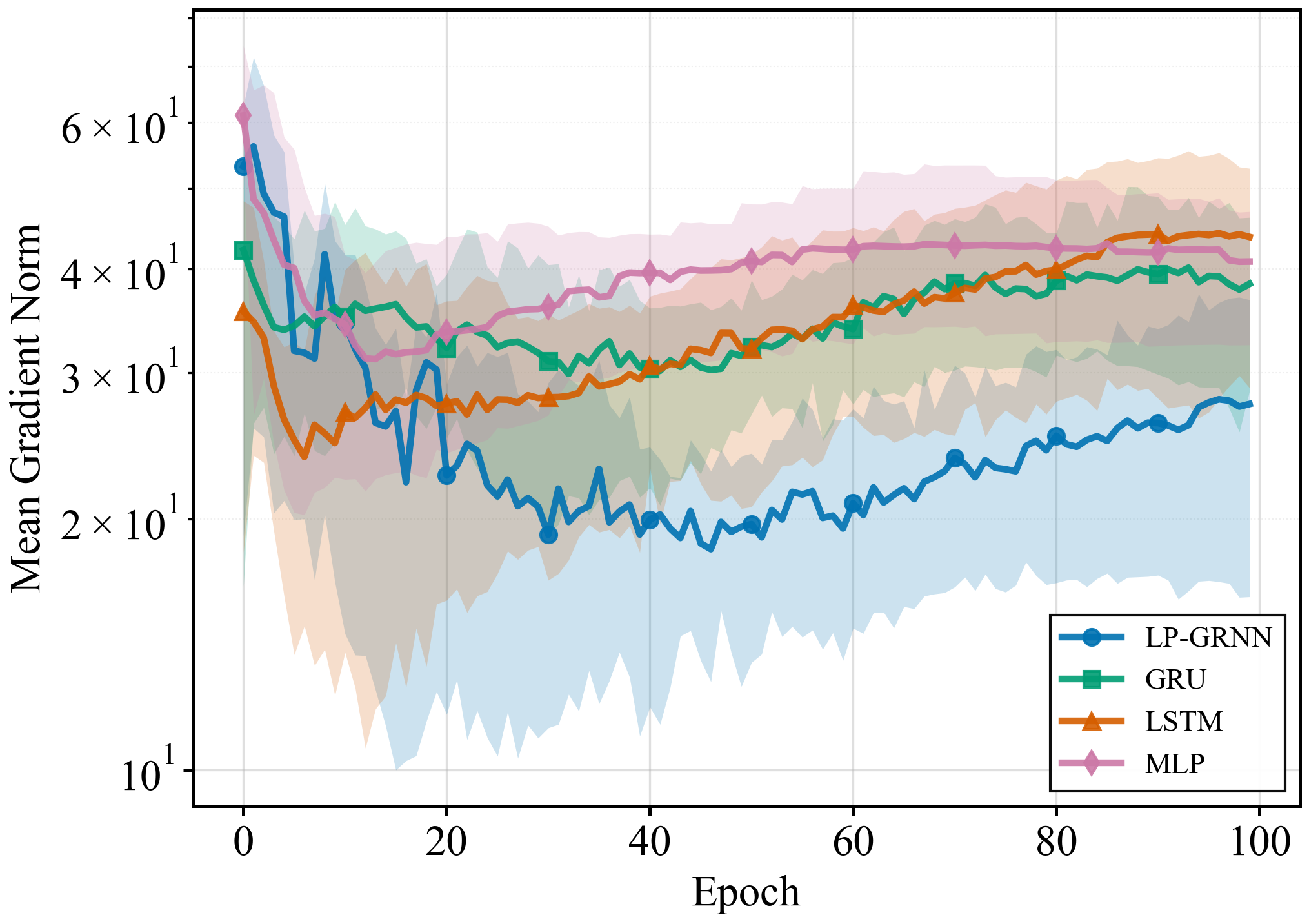}
\caption{Mean gradient norms across training epochs.}\label{fig:gradient_norm}
\end{figure}

Furthermore, we assessed whether the simplified affine gating structure induces vanishing gradients. The effective memory length, measured via Backpropagation Through Time (BPTT) gradient sensitivity, is reported in Table \ref{tab:dynamics_metrics}. The LP-GRNN retains an effective memory of 46.0 steps, which is comparable to the vanilla GRU's 49.3 steps. This confirms that the LP-GRNN achieves SNOF-compatibility without sacrificing the ability to capture long-term dependencies.

\begin{table}[htb]
    \centering
    \caption{Stability and Dynamics Metrics (Mean over Trials)}
    \label{tab:dynamics_metrics}
    \begin{tabular}{lcc}
        \hline
        \textbf{Method} & \textbf{Mean Grad. Norm} & \textbf{Eff. Memory (Steps)} \\
        \hline
        MLP     & 43.37 & -\\
        LSTM    & 32.54 & 48.7 \\
        GRU     & 33.46 & 49.3 \\
        \textbf{LP-GRNN} & \textbf{24.45} & \textbf{46.0} \\
        \hline
    \end{tabular}
\end{table}

Collectively, the analysis confirms that the LP-GRNN exhibits well-behaved training dynamics, effectively mitigating both exploding and vanishing gradients despite the architectural constraints.

\subsection{Illustrative Example}
The example detailed in this paper is closely aligned to the one demonstrated in \cite{10394670}. It is based on the linear plant proposed by \cite{Tan2005AnalysisAC}. The linearized model itself is derived from a nonlinear boiler plant, as described in \cite{strm1987DynamicMF}. The system in \cite{Tan2005AnalysisAC} has been slightly modified to align with the current problem. We model the output signal $y_3$, which is therefore removed from the output. Instead we directly output the state $x_3$. This leads to the following plant system:
\begin{equation}
\begin{aligned} 
   &A^{(p)} = 
    \left[\begin{smallmatrix*}[r]
    -0.0025 & 0 & 0 \\
    0.0694 & -0.1 & 0 \\
    -0.0067 & 0 & 0
    \end{smallmatrix*}\right], \quad
    B^{(p)} = \left[\begin{smallmatrix*}[r]
    0.9 & -0.349 & -0.15 \\
    0 & 14.155 & 0 \\
    0 & -1.398 & 1.659
    \end{smallmatrix*}\right], \\
   &C_{\Delta}^{(p)} = I_{v}, \quad
   D_{\Delta}^{(p)} = 0_{v \times s}.
\end{aligned}
\end{equation}

The plant is discretized using the zero-order hold method and excited by random input signals bounded by the actuator constraints as described in \cite{strm1987DynamicMF}. Initial values correspond to operating point \#4 from \cite{Tan2005AnalysisAC}. For the neural network we use the plant states and the plant input signals as the input features. The target signal $y_3$ is recorded as ground truth for supervised learning. The architecture of the neural network is based on our LP-GRNN. We use one hidden layer with three hidden neurons and a linear output layer. This results matrices stated in \eqref{eq:lp-grnn-matrices}

\begin{figure*}[!t]
    \normalsize
    \begin{equation}\label{eq:lp-grnn-matrices}
    \begin{aligned}
    A^{(n)} &= I_n - \alpha, \quad
    B^{(n)}_{p} = \left[\begin{smallmatrix} \alpha & 0_{n\times n} \end{smallmatrix}\right], \quad
    B^{(n)}_{u} = 0_{n\times m} , \quad
    \beta_x^{(n)} = 0_{n\times 1},  \quad
    C^{(n)}_{q} \thickapprox
    \left[\begin{smallmatrix*}[r]
    0.1427 & -0.2482 & 0.1805 \\
    0.6515 & -0.2854 & -0.4809 \\
    0.3640 & -0.1619 & 0.5506 \\
    0.2505 & -0.0978 & 0.1890 \\
    0.2372 & 0.2452 & 0.3117 \\
    -0.3550 & 0.2572 & 0.0626 \\
    \end{smallmatrix*}\right], \\
    D^{(n)}_{qp} &\thickapprox \left[\begin{smallmatrix*}[r]
    0 & 0 & 0 & -0.3197 & -0.4381 & 0.1456 \\
    0 & 0 & 0 & -0.2301 & -0.2483 & 0.1766 \\
    0 & 0 & 0 & 0.1071 & 0.1560 & -0.2065 \\
    0 & 0 & 0 & 0 & 0 & 0 \\
    0 & 0 & 0 & 0 & 0 & 0 \\
    0 & 0 & 0 & 0 & 0 & 0 \\
    \end{smallmatrix*}\right], \quad
    D^{(n)}_{qu} \thickapprox \left[\begin{smallmatrix*}[r]
    -0.0130 & 0.3061 & 0.8774 & -0.0251 & -0.1597 & -0.0432 \\
    -0.1069 & 0.1355 & 0.7464 & 0.1768 & 0.0714 & 0.2118 \\
    -0.0239 & 0.1095 & -0.0679 & 0.3225 & -0.0068 & -0.1522 \\
    0.0064 & 0.0097 & -0.5797 & -0.0085 & 0.0561 & 0.0652 \\
    -0.1427 & -0.1356 & -0.5616 & 0.0216 & -0.1105 & -0.0130 \\
    -0.1405 & 0.2245 & 0.2678 & -0.1161 & -0.1553 & -0.2022 \\
    \end{smallmatrix*}\right], \quad
    \beta_q^{(n)} \thickapprox \left[\begin{smallmatrix*}[r]
    -0.6407 \\ -0.3345 \\ 0.5554 \\ -0.0465 \\ 0.2233 \\ -0.1886 \\
    \end{smallmatrix*}\right], \\
    C^{(n)}_{y} &\thickapprox \left[\begin{smallmatrix*}[r]
    0.2879 \\ 0.2610 \\ 0.1536
    \end{smallmatrix*}\right]^T, \quad
    D^{(n)}_{yp} \thickapprox \left[\begin{smallmatrix*}[r]
    1.2667 \\ 1.2937 \\ 1.4011 \\ 0 \\ 0 \\ 0
    \end{smallmatrix*}\right]^T, \quad
    D^{(n)}_{yu} = 0_{1\times m}, \quad
    \beta_o^{(n)} \thickapprox 0.3077, \quad \alpha = \left[\begin{smallmatrix*}[r]
    0.4503 & 0 & 0 \\
    0 & 0.4090 & 0 \\
    0 & 0 & 0.5541 \\
    \end{smallmatrix*}\right].
    \end{aligned}
    \end{equation}
    \hrulefill
    \vspace*{4pt}
\end{figure*}

Between plant and virtual sensor, as well as between virtual sensor and controller, scaling is implemented as outlined in Section \ref{sec:pre}. A min-max scaler is utilized for this illustrative example. The factors $\sigma_i$, $\theta_i$, $\sigma_o$, and $\theta_o$ are calculated based on the minimum and maximum values present in the training data:

\begin{equation}
\begin{aligned}
\sigma_i &\thickapprox
\left[\begin{smallmatrix*}[c]
0.0023 & 0 & 0 & 0 & 0 & 0\\
0 & 0.0027 & 0 & 0 & 0 & 0 \\
0 & 0 & 0.0004 & 0 & 0 & 0 \\
0 & 0 & 0 & 1 & 0 & 0\\
0 & 0 & 0 & 0 & 1 & 0 \\
0 & 0 & 0 & 0 & 0 & 1
\end{smallmatrix*}\right], \quad
\theta_i \thickapprox
\left[\begin{smallmatrix*}[r]
0.3279 \\
0.1920 \\
0.5734 \\
0 \\
0 \\
0
\end{smallmatrix*}\right], \\
\sigma_o &= 10, \quad 
\theta_o = -5
\end{aligned}
\end{equation}

Finally, a continuous PI-controller from \cite{10394670} automatically designed by the tuning algorithm from \cite{MATLAB:2019} is used. After its transformation to state-space and a discritization via the zero-order method, its matrices are given by

\begin{equation}
\begin{aligned}
A^{(c)} &= 
\left[\begin{smallmatrix}
1 & 0 & 0 \\
0 & 1 & 0 \\
0 & 0 & 1 \\
\end{smallmatrix}\right], \quad
B^{(c)}_{u} \thickapprox
\left[\begin{smallmatrix*}[r]
0.2 & 0 & 0 \\
0 & 6.3 & 0 \\
0 & 0 & 0.4 \\
\end{smallmatrix*}\right] \times 10^{-2}, \\
C^{(c)}_{q} &\thickapprox
\left[\begin{smallmatrix*}[r]
0.1 & 0 & 0 \\
0 & 4.1 & 0 \\
0 & 0 & 0.6
\end{smallmatrix*}\right] \times 10^{-2}, \quad
D^{(c)}_{qu} \thickapprox 
\left[\begin{smallmatrix*}[r]
0.1 & 0 & 0 \\
0 & 0.8 & 0 \\
0 & 0 & 3.2
\end{smallmatrix*}\right] \times 10^{-2}.
\end{aligned}
\end{equation}

As in \cite{10394670}, the controller output is constrained by the nonlinear operator detailed in \eqref{eq:CSNOF}, with minimum values $y_{\text{min}}^{(c)} = 0_s$ and maximum values $y_{\text{max}}^{(c)} = 1_s$. The SNOF, as constructed according to \eqref{eq:CPNNSNOF}, can be shown to be well-formed, as it satisfies all conditions outlined in Section \ref{sec:pre}.

The system can be assessed for global asymptotic stability utilizing the Theorem \ref{th:stabtheorem}.The convex optimization problem was solved in Python using CVXPY \cite{JMLR:v17:15-408} with the SCS solver \cite{o2016conic}, configured as follows:
\begin{itemize}
    \item $\varepsilon_{\text{abs}} = 1\times10^{-5}$
    \item $\varepsilon_{\text{rel}} = 1\times10^{-5}$
    \item $\text{max\_iters} = 10^5.$
\end{itemize}

Thus, within the numerical tolerances of the solver, the system is confirmed to be globally asymptotically stable. The closed-loop performance is visualized in Fig.~\ref{fig:closed_loop_response}. 

\begin{figure*}[htb]
  \centering
  \includegraphics[width=2\columnwidth]{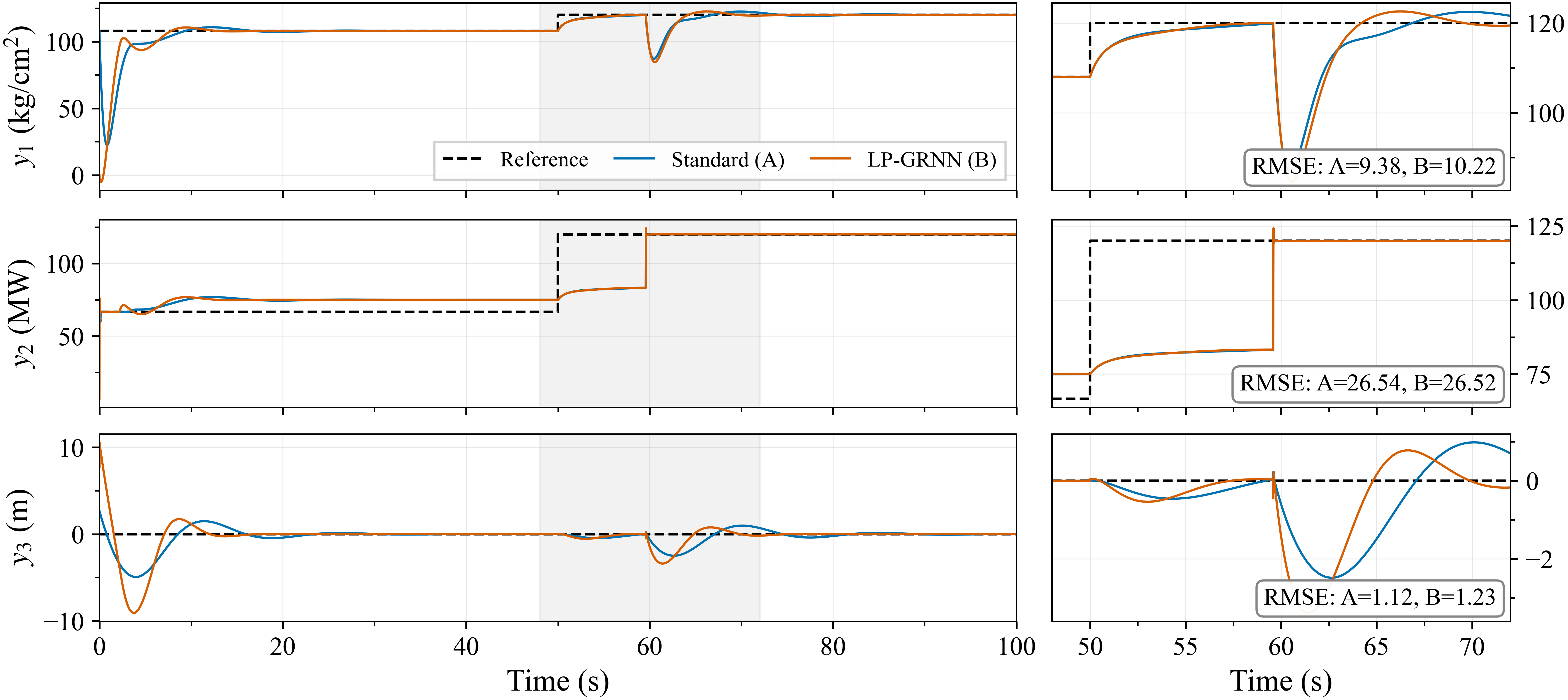}
  \caption{Closed-loop response of the linearized boiler system comparing standard control (A) with LP-GRNN virtual sensor-based control~(B).Left column: full simulation with setpoint changes at $t=50$\,s. Right column: magnified view of the transient response. Output $y_3$ (drum level) is estimated by the LP-GRNN virtual sensor in configuration (B), while $y_1$ and $y_2$ are directly measured in both configurations.}
  \label{fig:closed_loop_response}
\end{figure*}

\begin{table}[htb]
    \centering
    \caption{Step Response Performance Metrics}
    \label{tab:step_response_metrics}
    \begin{tabular}{llrr}
        \hline
        \textbf{Output} & \textbf{Metric} & \textbf{Standard (A)} & \textbf{LP-GRNN (B)} \\
        \hline
        $y1$ & RMSE & 9.3831 & 10.2226 \\
        & IAE & 104.59 & 106.81 \\
        & Overshoot (\%) & 20.67 & 21.60 \\
        & SS Error & 1.4705 & 1.4020 \\
        \hline
        $y2$ & RMSE & 26.5367 & 26.5204 \\
        & IAE & 367.14 & 367.18 \\
        & Settling Time (s) & 9.60 & 9.62 \\
        & Overshoot (\%) & 7.18 & 7.90 \\
        & SS Error & 0.0019 & 0.0014 \\
        \hline
        $y3$ & RMSE & 1.1173 & 1.2350 \\
        & IAE & 15.91 & 15.23 \\
        & SS Error & 0.5711 & 0.4728 \\
        \hline
    \end{tabular}
\end{table}

Fig.~\ref{fig:closed_loop_response} compares the closed-loop behavior of both control configurations. The LP-GRNN-based system (B) achieves tracking performance comparable to the standard controller (A), with similar RMSE values for all outputs. Notably, Table \ref{tab:step_response_metrics} shows that the virtually sensed output $y_3$ shows slightly improved steady-state accuracy (SS error: 0.47 vs.\ 0.57), demonstrating that the closed-loop system containing the virtual sensor does not compromise control quality and exhibits dynamics and stability comparable to the original configuration.

\section{Conclusion}\label{sec:conclusion}
This paper addresses the certification of stability for control systems that incorporate gated recurrent–network–based virtual sensors. We demonstrate that the Hadamard gate in GRU/LSTM cells violates the path‐independent storage‐function requirement, making standard Lyapunov tests inapplicable. To overcome this, we propose the LP-GRNN, which replaces the multiplicative gate with an affine update while preserving temporal memory. The LP-GRNN admits an exact Standard Nonlinear Operator Form with static, component-wise $\tanh$ nonlinearities that satisfy sector and slope bounds. Thus, we could formulate a well-posed SNOF for the closed-loop system combining the LP-GRNN, the plant, and a constrained controller and apply the global-asymptotic-stability theorem of \cite{kim2009robust}.

Empirical evaluation shows that the LP-GRNN matches the prediction accuracy of standard GRUs/LSTMs on a time-series benchmark, and a linearized boiler example illustrates how the resulting LMI can be solved numerically to verify stability. The main limitations are a reduction in architectural flexibility—potentially requiring larger LP-GRNNs for complex tasks—and quadratic growth of the LMI with hidden-state size, which may hinder scalability. Additionally, the current analysis assumes ideal conditions without delays or sampling mismatches, which would require extensions for practical deployment.

For future work, the approach presented in this article can be also applied to stacked LP‑GRNNs. One could also swap the fixed PI loop for an LP‑compatible, learned controller and analyze the joint plant‑sensor‑controller system in the same way.  Finally, work on the quadratic LMI scalability issue is needed. One could exploit sparsity—using chordal decomposition \cite{XUE2024111487} or use block‑diagonal factorization so each layer is verified with a much smaller LMI that can be solved in parallel \cite{8264648}.

\section*{Acknowledgment}
This work was not supported by any organization.

\bibliographystyle{IEEEtran}
\bibliography{refs}

\end{document}